\documentclass{cta-author}

{}
\newtheorem{corollary}{Corollary}{}
\newtheorem{remark}{Remark}{}
\newtheorem{proposition}{Proposition}{}

\usepackage{subfigure}

\def\begcen{\begin{center}}
\def\endcen{\end{center}}

\newcommand{\col}{\mbox{col}}

\def\calf{{\mathcal F}}

\def\calj{{\mathcal J}}

\def\cald{{\mathcal D}}
\def\calo{{\mathcal O}}
\def\calz{{\mathcal Z}}
\def\hal{\frac{1}{2}}

\def\calh{{\mathcal H}}

\def\L2e{{\mathcal L}_{2e}}

\def\rea{\mathbb{R}}

\def\et{\epsilon_t}

\def\l2{{\mathcal L}_2}
\def\l2e{{\cal L}_{2e}}

\def\rea{\mathbb{R}}

\def\iab{i_{\alpha\beta}}
\def\vab{v_{\alpha\beta}}

\def\biab{\bar{i}_{\alpha\beta}}
\def\hpf{\texttt{HPF}}
\def\lpf{\texttt{LPF}}
\def\hgrad{\mathcal{G}_{\texttt{grad}}}

\def\begequarr{\begin{eqnarray}}
\def\endequarr{\end{eqnarray}}
\def\begequarrs{\begin{eqnarray*}}
\def\endequarrs{\end{eqnarray*}}
\def\begarr{\begin{array}}
\def\endarr{\end{array}}
\def\begequ{\begin{equation}}
\def\endequ{\end{equation}}
\def\lab{\label}
\def\begdes{\begin{description}}
\def\enddes{\end{description}}
\def\begenu{\begin{enumerate}}
\def\begite{\begin{itemize}}
\def\endite{\end{itemize}}
\def\endenu{\end{enumerate}}

\def\lef[{\left[\begin{array}}
\def\rig]{\end{array}\right]}
\def\qed{\hfill$\Box \Box \Box$}
\def\begcen{\begin{center}}
\def\endcen{\end{center}}
\def\begrem{\begin{remark}\rm}
\def\endrem{\end{remark}}

\def\TIE{{\it IEEE Trans. Industrial Electronics}}
\def\TIA{{\it IEEE Trans. Industry Applications}}

\def\EJC{{\it European J. of Control}}
\def\IJC{{\it Int. J. of Control}}

\def\SCL{{\it Systems \& Control Letters}}
\def\AUT{{\it Automatica}}

\def\CST{{\it IEEE Trans. Control Systems Technology}}

\def\begmat#1{\begin{bmatrix}#1\end{bmatrix}}
\def\begali#1{\begin{align}{#1}\end{align}}

\def\cfm#1{\textnormal{\textsc{\textbf{\textcolor{cBlue}{#1}}}}}

\usepackage{color}
\newcommand{\blue}[1]{{\color{blue} #1}}

\definecolor{cBlue}{rgb}{0.2,0.6,1}

\begin{document}


\title{\textnormal{\textsc{\textbf{\textcolor{cBlue}{A New Signal Injection-based Method for Estimation of Position in Interior Permanent Magnet Synchronous Motors}}}}}

\author{\rm \au{Bowen Yi$^{1,2}$}, \au{Slobodan N. Vukosavi\'c$^{3}$}, \au{Romeo Ortega$^{2,4}$}, \au{Aleksandar M. Stankovi\'c$^{5}$}, \au{Weidong Zhang$^{1\corr}$}}

\address{\add{1}{Department of Automation, Shanghai Jiao Tong University, Shanghai 200240, China}
\add{2}{Laboratoire des Signaux et Syst\`emes, CNRS - CentraleSup\'elec, Gif-sur-Yvette 91192, France}
\add{3}{Electrical Engineering Department, University of Belgrade, Belgrade 11000, Serbia}
\add{4}{Department of Control Systems and Informatics, ITMO University, St. Petersburg, Russia}
\add{5}{Department of Electrical Engineering and Computer Science, Tufts University, Medford, MA 02155, USA}
\email{wdzhang@sjtu.edu.cn}
}

\begin{abstract}
Several heuristic procedures to estimate the rotor position of interior permanent magnet synchronous motors (IPMSM) via signal injection have been reported in the applications literature, and are widely used in practice. These methods, based on the use of linear time invariant (LTI) high-pass/low-pass filtering, are instrumental for the development of sensorless controllers. To the best of our knowledge, no theoretical analysis of these methods has been carried out. The objectives of this note, are (i) to invoke some recent work on the application of averaging techniques for injection-based observer design to develop a theoretical framework to analyze the LTI filtering used in sensorless methods, and (ii) to propose a new method that, on one hand, ensures an improved accuracy and, on the other hand, can be related with the current filtering technique. An additional advantage of the new method is that it relies on the use of linear operators, implementable with simple computations. The effectiveness of the proposed scheme is assessed by experiments on an IPMSM platform driven by a 521 V DC bus with a 5-kHz PWM.
\end{abstract}

\maketitle

\section*{\cfm{\large Nomenclature}}

\vspace{-1cm}

\begin{table}[h]
\normalsize
\textbf{Symbols}\\

\vspace{0.2cm}

\begin{tabular}{ll}
$\alpha-\beta$ & Stationary axis reference frame quantities\\
$d-q$ & Synchronous axis reference frame quantities \\
$n_p$ & Number of pole pairs \\
$R_s$    & Stator resistance [$\Omega$]\\
$\omega$ & Angular velocity [rad/s] \\
$\Phi$ & Magnetic flux [Wb] \\
$J$ & Drive inertia [kg$\cdot\text{m}^2$] \\
$T_L$ & Load torque [N$\cdot$m]\\
$f$ & Friction constant \\
$\theta$ & Rotor flux angle [rad] \\
$L_d,L_q$ & $d$ and $q$-axis inductances [H] \\
$v,i$ & Stator voltage and current [V, A] \\
$\omega_h$ & Angular frequency of injection signal [rad/s]\\
$\varepsilon$ & Period of injection signal ($\varepsilon= {2\pi \over \omega_h}$) [s] \\
$V_h$ & Amplitude of injection signal [V]\\
${\tt HPF}$ & High-pass filter \\
${\tt LPF}$ & Low-pass filter \\
${\tt BIBO}$ & Bounded-input bounded-output \\
$|\cdot|$ & Euclidean norm \\
$s$ & {Laplace transform symbol}\\
$p$ & {Differential operator $p={d\over dt}$}\\
$y(t)=\calh[u(t)]$ & ${\tt BIBO}$ operator $\calh$ acting on the input signal\\
& $u(t)$ to generate the output $y(t)$.\\
$y_v$ & Virtual output\\
$\iab$ & $[i_\alpha, i_\beta]^\top$ \\
$\vab$ & $[v_\alpha,v_\beta]^\top$ \\
$I,~\calj$ & Identity matrix on $\rea^{2\times 2}$ and $\begin{bmatrix} 0 & -1 \\1 & 0 \end{bmatrix}$
\\
\end{tabular}
\end{table}


\begin{table}[h]
\normalsize
\textbf{Superscripts}\\

\begin{tabular}{ll}
$r$ & Actual reference frame\\
$\hat{r}$ & Estimated reference frame\\
$r^\star$ & Reference value \\
$r^h$ & High-frequency component \\
$r^\ell$ & Low-frequency component \\
$ \vab^C$ & Low-frequency control input
\end{tabular}
\end{table}

%
\section{\cfm{\large Introduction}}
%

Permanent magnet synchronous motors (PMSMs) are widely used in industrial applications because of their superior power density and high efficiency.  One of the more---practically and theoretically---challenging open problems for PMSMs, is the design of controllers without rotational sensors, the so-called sensorless control. Two different types of sensorless control methodologies are currently being used in practice. The first one is a model-based method, which is known in applications as the back-emf or flux-linkage estimation. In this method, the fundamental components of the electrical signals are used to design a back-emf observer or a {\em flux observer} \cite{ACAetal,BERPRA,BOBetal,BOBetalijc,MALetal,MATtie,ORTetalcst,POUetal,TILetal}. The second one is saliency-tracking-based method, in which information is extracted from the high-frequency components of stator currents via high-frequency {\em signal injections} \cite{JANetal,HOLtie,JEBetal,JANetaltia}.

It is well-known that, because of the loss of observability at standstill \cite{POUetal}, observer-based methods cannot be used in low-speed region \cite{JANetal}. On the other hand, the performance of the saliency tracking-based method, which utilizes the anisotropy due to rotor saliency and/or magnetic saturation, is not degraded at low speeds. In this paper, we address the problem of position estimation for PMSMs at \emph{low speeds or standstill}, using a signal injection based method.

There are two kinds of PMSMs, surface (SPMSM) or interior (IPMSM), the difference been the location of the permanent magnets, either on the surface of the rotor, for the former, or buried in the cavities of the rotor core, for the latter. There are several technological reasons why IPMSMs are more convenient in applications than SPMSMs, see \cite[Table 6.2]{NAMbook}. On the other hand, the magnetic characteristics of IPMSM, and consequently the dynamic model is, far more complicated than the one of SPMSM---see the discussion in \cite[Subsection 6.2.2]{NAMbook} and \cite[Section VI]{ORTetalcst}. As a matter of fact, because of the inability to deal with this complex dynamics, the overwhelming majority of papers published by the control community in sensorless control {operating in middle or high speed regions} are concerned with SPMSMs, see {\em e.g.}, \cite{BERPRA,BOBetal,BOBetalijc,MALetal,POUetal,TILetal}---with \cite{CHOetal,ORTetalauto} a notable exception. {In contrast, the saliency-tracking methods at low speeds {have} received more attentions from the applied journals in the context of IPMSMs.}

Due to the {rotor} saliency, the signal injection method is more efficiently applied in IPMSMs than in SPMSMs. For this reason, we consider in this paper IPMSMs.  In the last two decades, signal injection-based approaches have been successfully applied, in a heuristic manner, in various applications, some of them reported in application journals \cite{GONZHU,HOLtie,JANetal,JANetaltia,JEOetal,NAMbook}. The classical approach is, first, to inject high frequency probing signals into the motor terminal; then, extract the high-frequency components of the stator currents to get position estimates. Besides the question of the choice of the injection signal, and its mode of application, the key problem is the signal processing of the measured stator currents to extract the desired information of the mechanical coordinates. This task is usually achieved via the combination of linear time invariant (LTI) high pass-filters (HPFs) and low-pass filters (LPFs) \cite{NAMbook}---an approach which is justified by a series of technique-oriented practical considerations, hard to fit into a rigorous theoretical framework. To the best of our knowledge, no theoretical analysis of these heuristic methods has been reported in the literature.

Our contributions in this paper are threefold.

\begite
\item Provide a theoretical framework for the analysis of conventional LTI filtering methods used in injection-based sensorless control of IPMSMs.

    \vspace{0.1cm}

\item Propose a new method for the extraction of the information on the mechanical coordinates that, using the aforementioned framework, is shown to be superior to the existing {LTI filtering technique}. An important aspect in this point is that the increase in computational complexity with respect to the current HPF/LPF practice should be negligible.

    \vspace{0.1cm}

\item Prove that the new proposed method admits an HPF/LPF interpretation. This is an important issue, since it shows the connection---and downwards-compatibility---of the new method with standard industrial practice.
\endite

Towards this end, we rely on the recent work of \cite{COMetalacc,COMetal,JEBetal,YIetalscl} where, invoking averaging techniques, rigourous theoretical analysis of injection-based methods for observer design has been carried out. The importance of disposing of rigorous analytic results can hardly be overestimated, since it allows, on one hand, to carry out a quantitative performance assessment while, on the other hand, it provides guidelines to make more systematic and simplify the parameter tuning procedure. {We underscore that the IPMSM model adopted in the paper is widely accepted by the drives community, since it precisely describes the behavior of the machine in the absence of magnetic saturation, a phenomenon that is conspicuous by its absence when the load in the motor is within the normal operating range.}

The remainder of paper is organized as follows. In Section \ref{sec2}, we recall the mathematical model of IPMSMs and formulate the problem of estimation of position using signal injection. Section \ref{sec3} discusses the classical frequency-based accuracy analysis of the position estimators used for the conventional methods, and highlights their theoretical limitations. In Section \ref{sec4} the new method is proposed, and then some comparisons and similarities with the conventional methods are given in Section \ref{sec5}. Simulation and experimental results are given in Section \ref{sec6}. The paper is wrapped-up with some concluding remarks in Section \ref{sec7}.

\noindent {\em Caveat.} An abridged version of this paper has been presented in \cite{YIccta}.

%
\section{\cfm{\large Model and Problem Formulation}}
\label{sec2}
%
The voltage equations of the IPMSM in the stationary frame are given by \cite{NAMbook}
\begin{equation}
\label{volt_model}
\begin{aligned}
\vab  = \Big[ R_sI + L(\theta){p} - 2n_p\omega L_1 Q(\theta) \calj \Big]\iab
           + {n_p\omega}\Phi \left[\begin{aligned} - \sin\theta \\ \cos\theta \end{aligned}\right],
\end{aligned}
\end{equation}
where we define the mappings
$$
\begin{aligned}
L(\theta)& := L_0I + L_1Q(\theta)\\
Q(\theta)& := \left[\begin{aligned} &\cos2\theta & \sin2\theta \\& \sin2\theta & - \cos2\theta \end{aligned}\right],
\end{aligned}
$$
with the averaged inductance $L_0$ and the inductance difference value $L_1$ as
$$
L_0  :={1\over2} (L_d+L_q), \quad L_1  := {1\over2} (L_d - L_q).
$$

The stationary model \eqref{volt_model}, together with the mechanical dynamics, can be expressed in the standard state-space form as follows.
\begequ
\label{sal_pmsm1}
    \begin{aligned}
        L(\theta) {d  \over dt}{\iab}  & =  F(\iab,\theta,\omega) +   \vab\\
        {d  \over dt}{\theta} & = n_p \omega \\
        J{d  \over dt}{\omega} & = n_p \Phi (i_\beta \cos\theta - i_\alpha \sin\theta) - f\omega - T_L,
    \end{aligned}
\endequ
where we define the mapping
$$
F (\cdot)  :=
 \big(2n_p \omega L_1 Q(\theta)\calj - R_s I \big)\iab   + {n_p\omega} \Phi
\left[\begin{aligned}  \sin\theta \\ -\cos\theta \end{aligned}\right].
$$

\noindent {\bf Problem Formulation} (Position Estimation via Signal Injection): Assume there is a stabilizing controller operator $\Sigma_C$ measuring only $\iab$, and define its output as
$$
 \vab^C(t) :=\Sigma_C[\iab(t)].
$$
Inject a high-frequency signal to one axis of the control voltage, say, the $\alpha$-axis, that is,
\begin{equation}
\label{siginj}
\vab = \vab^C +
\begin{bmatrix}
V_h\sin \omega_h t \\ 0
\end{bmatrix},
\end{equation}
where $\omega_h:={2\pi \over\varepsilon}$, with $\varepsilon > 0$ {\em small}, and $V_h>0$. The problem is to define an operator
$$
\Sigma_E:\iab \mapsto \hat \theta
$$
such that
\begequ
\lab{estacc}
\limsup_{t\to\infty} \big| \hat{\theta}(t) - \theta(t)  \big|  \le \mathcal{O}({\varepsilon}),
\endequ
where $\mathcal{O}$ is the uniform big O symbol.\footnote{\rm That is, $f(z,\varepsilon)=\mathcal{O}(\varepsilon)$ if and only if $|f(z,\varepsilon)|\le C\varepsilon$, for a constant $C$ independent of $z$ and $\varepsilon$. Clearly, $\calo(1)$ means the boundedness of a signal.}

It is well-known that high frequency probing signals have almost no effect on the motor mechanical coordinates. However, due to the rotor saliency, it induces different high-frequency responses in the $\alpha$- and $\beta$-axes currents. This fact provides the possibility to recover the angle from the high-frequency components of stator currents.

%
\section{\cfm{\large Frequency Decomposition and Quantitative Analysis of Conventional Methods}}
\label{sec3}
%
In this section, we give the analysis of frequency decomposition of the stator currents $\iab$, which is instrumental for the design and analysis of position estimators.

\subsection{Conventional Frequency Analysis}
First we recall the conventional frequency decomposition in the technique-oriented literature, which relies on the {\em ad-hoc} application of the superposition law \cite{JANetal,NAMbook}. That is, suppose the electrical states {consist} of high-frequency and low-frequency components as
$$
(\cdot)_{\alpha\beta} =(\cdot)_{\alpha\beta}^h + (\cdot)_{\alpha\beta}^\ell.
$$
If $\omega \approx 0$, the current responses can be separated as
\begequ
\label{separation}
\vab^\ell + \vab^h = \big( R_sI + L(\theta){p} \big) (\iab^\ell + \iab^h).
\endequ
For the approximated high-frequency model
$$
\vab^h \approx L(\theta) {p} \iab^h,
$$
neglecting the stator resistance, the angle $\theta$ can be regarded as a {\em constant}, thus the high-frequency response contains the information of $\theta$, namely, for the input \eqref{siginj} we have
$$
\begin{aligned}
i_\alpha^h & = {V_h (L_0-L_1\cos2\theta)  \over  L_dL_q {p}} \big[ \sin\omega_h t\big]\\
i_\beta^h & = - {(V_h L_1\sin2\theta)   \over  L_dL_q {p}} \big[ \sin\omega_h t\big].
\end{aligned}
$$
{Substituting $p={d\over dt}$, we approximately get the high-frequency components of the stator {current} as
\begin{equation}
\label{hfc}
\iab^h =
{1\over \omega_h L_dL_q}
\left[ \begin{aligned}   & - L_1 \cos2\theta +  L_0\\   & - L_1 \sin2\theta \end{aligned} \right]
(-V_h\cos\omega_h t).
\end{equation}}

The derivation of the above high-frequency model is based on two assumptions, namely, the superposition law and the slow angular velocity $\omega\approx 0$, regarding which, the following remarks are in order.
\begin{itemize}
  \item[\bf R1] The dynamics \eqref{sal_pmsm1} is highly nonlinear. It is well-known that nonlinear systems ``mix'' the frequencies, making the superposition law not applicable. Although using the classical decomposition \eqref{separation} to estimate position may work in practice, it fails to reliably provide, neither a  framework for a  quantitative performance assessment, nor guidelines to tune parameters.
      \ \\ \
  \item[\bf R2] The assumption $\omega\approx 0$ implies that the decomposition above is applicable only at standstill or very low speeds.
\end{itemize}

\subsection{Frequency Analysis via Averaging}
\label{sec:32}
Averaging analysis provides a rigorous and elegant decomposition of the measured currents as follows. We refer the reader to \cite{HAL,SANbook} for the basic theory on averaging analysis. Applying averaging analysis, it is shown in \cite{JEBetal} that with $\omega_h >0$ large enough
\begequ
\label{ident1}
\iab = \biab + \varepsilon y_v S  + \mathcal{O} (\varepsilon^2),
\endequ
where, we defined the signal
\begequ
\lab{S}
S(t) := - { V_h \over 2\pi} \cos (\omega_h t),
\endequ
the (so-called) virtual output
\begequ
\lab{yv}
y_v := {1 \over L_dL_q}
\left[ \begin{aligned}   & - L_1 \cos2\theta +  L_0\\   & - L_1 \sin2\theta \end{aligned} \right],
\endequ
and $\biab$ is the current of the closed-loop system with $\vab=\vab^C$---that is, without signal injection. From \eqref{yv} it is clear that {the angular position $\theta$ can be recovered from the virtual output $y_v$, that is}
$$
{\theta} = {1 \over 2} \arctan \left\{ {y_{v_2}\over y_{v_1} - {L_0\over L_dL_q}}\right\}.
$$

Hence the position estimation problem is translated into the estimation of $y_v$. Towards this end, we notice that, from a frequency viewpoint, $\iab$ contains fundamental frequency component $\biab$ and high frequency component $\varepsilon y_v S$. It should be noticed that the high frequency term $\varepsilon y_v S$ coincides with the one in \eqref{hfc}, but the averaging analysis characterizes all the components in $\iab$ quantitatively.

It is natural, then, that to ``reconstruct" $y_v$---out of  measurements of $\iab$---we need to separate these components via some sort of HPF and LPF operations. This is the rationale underlying most of existing position estimators reported in the literature, see \cite{NAMbook} for a recent review.

\vspace{0.15cm}

\noindent{\bf R3} The tiny term $\mathcal{O} (\varepsilon^2)$ in \eqref{ident1} is caused by second-order periodic averaging analysis, which is concerned with solving a perturbation problem in a properly selected time scale.

\vspace{0.15cm}

\noindent{\bf R4} Rigorously, the function $\arctan$, adopted in the paper for convenience, should be replaced by the 2-argument arctangent function $\text{atan2}(\cdot,\cdot)$, the fact widely known in the drive community.
\subsection{Quantitative Results of Conventional Methods}
In \cite{NAMbook} the position estimation method, for low rotation speeds, shown in Fig. \ref{fig:block1} is proposed.
\begin{figure}[]
    \centering
    \includegraphics[width=8.5cm]{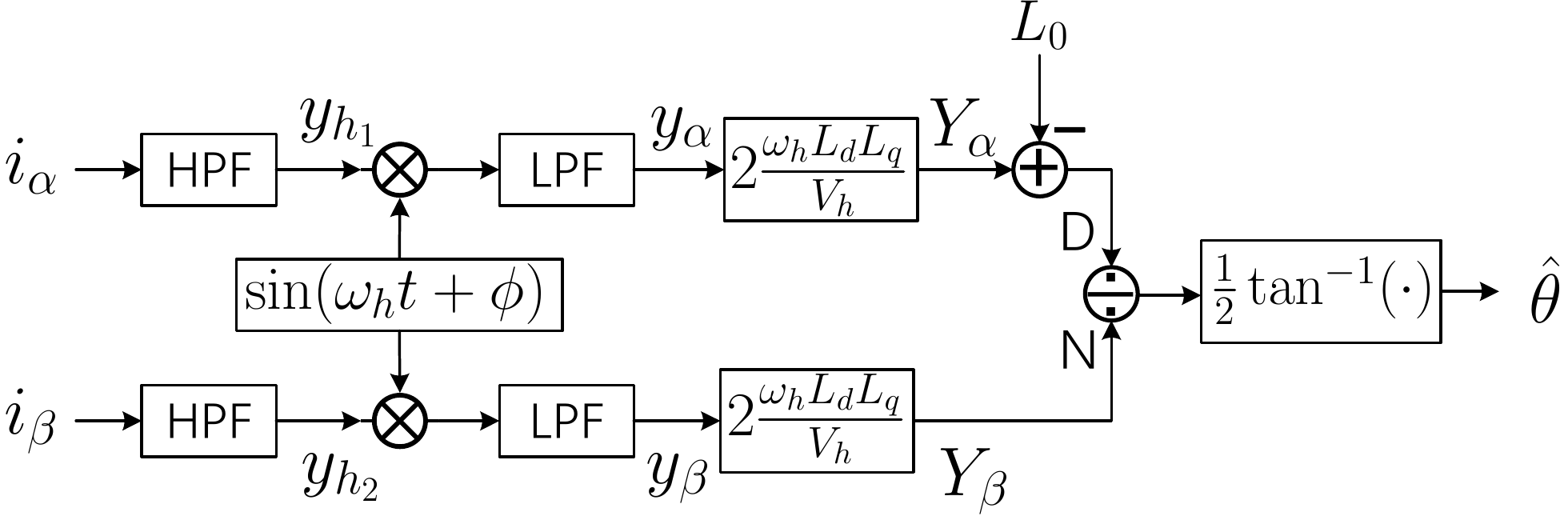}
    \caption{Block diagram of the conventional signal injection method \cite{NAMbook}}
    \label{fig:block1}
\end{figure}
To evaluate the performance of the classical method depicted in Fig. \ref{fig:block1}, without loss of generality, select the LTI filters {as the transfer function}\footnote{\rm Given a smooth signal $y(t)$ and its
Laplace transform $Y(s)$, the following operators in time domain can be written as $\hpf[y(t)] ={2p^2\over (\lambda_h +p)^2}[y(t)]$ and $\lpf[y(t)] = {\lambda_\ell \over \lambda_\ell + p}[y(t)]$. For brevity, we make a slight abuse of notation using both $\hpf[y(t)]$ and $\hpf(s)Y(s)$ below.}
  \begali{
\nonumber
{\tt HPF}(s) &= {2s^2 \over (\lambda_h + s)^2}\\
\lab{filters1}
{\tt LPF}(s) &= {\lambda_\ell \over \lambda_\ell + s},
}
with parameters
\begequ
\label{parameter-LTI}
\lambda_h = \omega_h,\; \lambda_\ell = \max\{\sqrt{\omega_h \omega_\star},1\}.
\endequ
The Bode diagrams of two filters are given in Fig. \ref{fig:bode_H/LPF} with $\omega_h=500$ and $\omega_\star=1$.

\begin{figure}[]
    \centering
    \includegraphics[width=4.25cm,height=4.2cm]{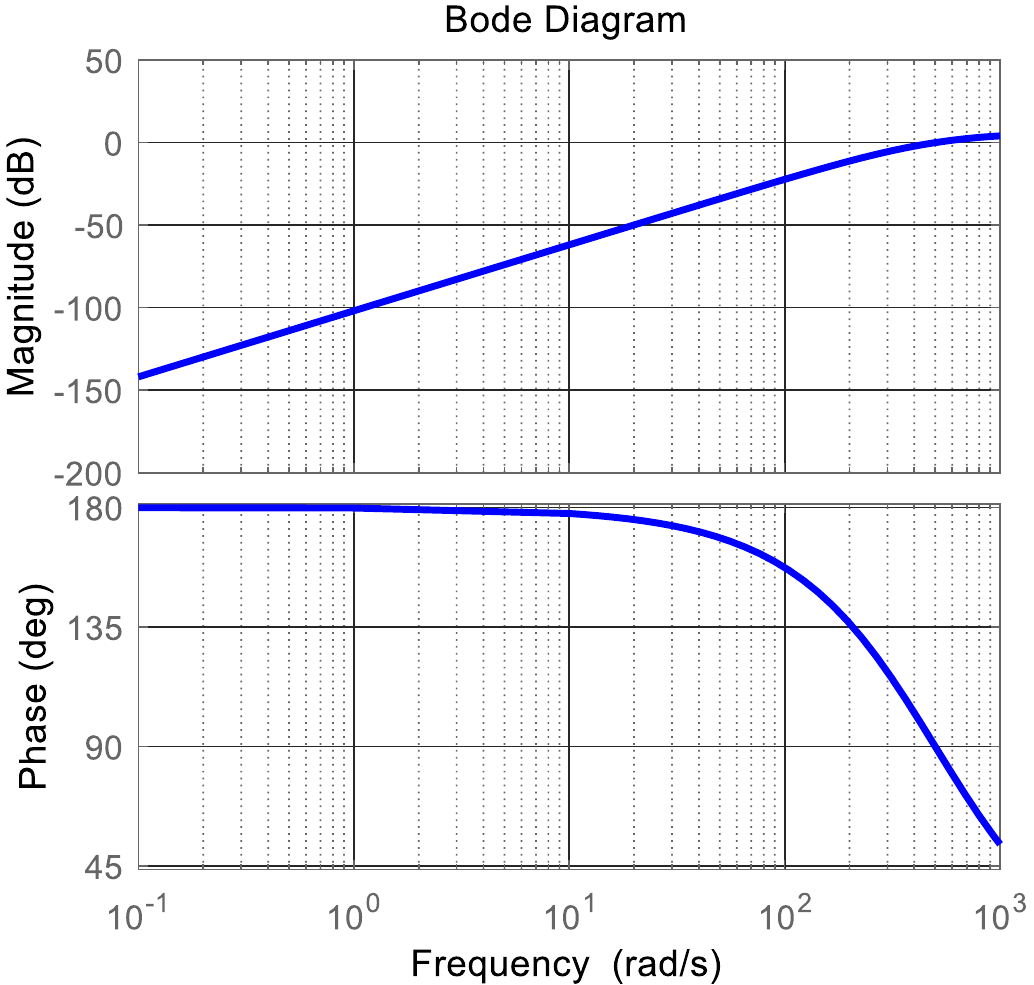}
    \includegraphics[width=4.25cm,height=4.2cm]{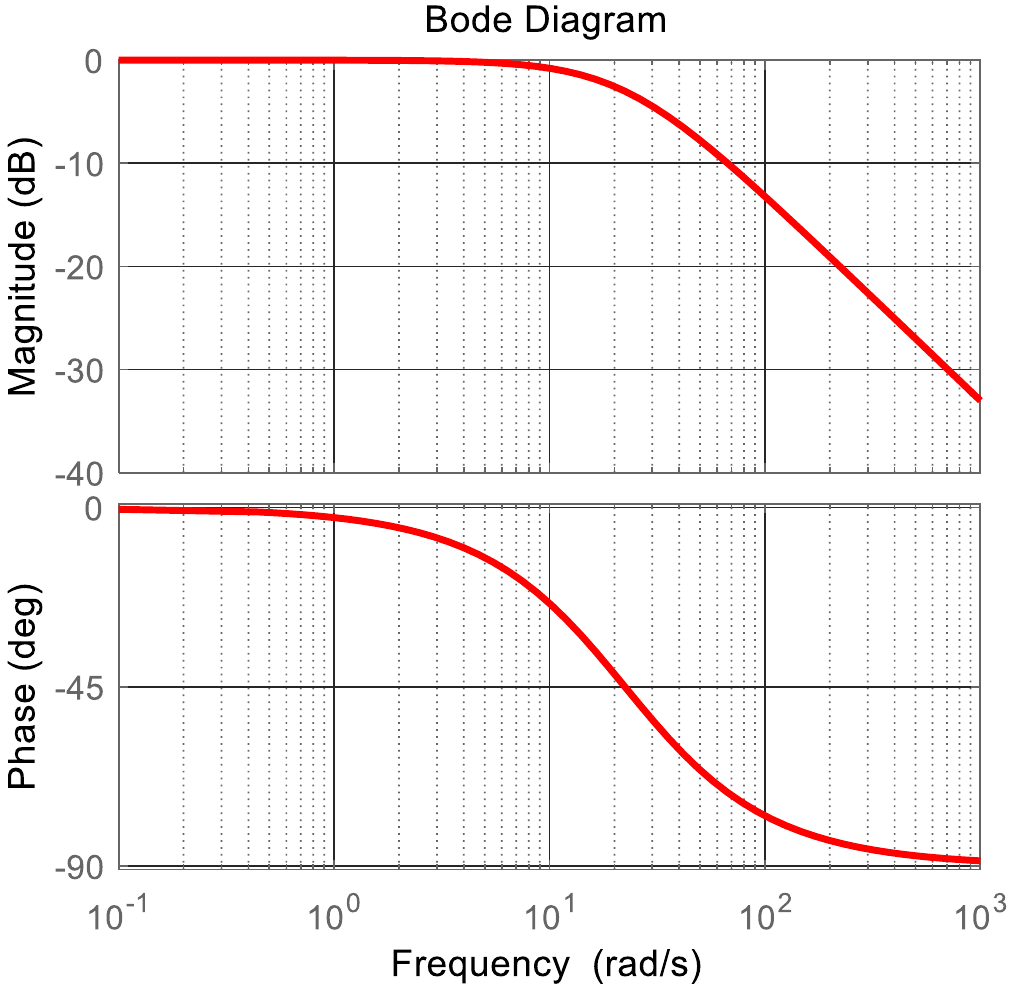}
    \caption{Bode diagram of the HPF/LPF \eqref{filters1} ($\omega_h=500,\;\omega_\star=1$)}
    \label{fig:bode_H/LPF}
\end{figure}

Applying averaging analysis at reduced speeds, and setting $\phi=0$, we have the following.

\begin{proposition}
\label{prop-conventional}\rm
For the IPMSM model \eqref{sal_pmsm1}, suppose the control $\vab^C$ guarantees all the states bounded, with the speed
$$
\big|(\bar\omega, \dot{\bar\omega}, \dot{v}_{\alpha\beta}^C) \big|
\le \ell_\omega
$$
for some constant $\ell_\omega$ independent of $\varepsilon $. If the filters are selected as \eqref{filters1}-\eqref{parameter-LTI}, then the signal processing procedure depicted in Fig. \ref{fig:block1}, namely,
\begequ
\label{sig_proc}
\begin{aligned}
   y_h & = \hpf[\iab]\\
   Y & = 2{\omega_h \over V_h}L_dL_q \lpf \big[y_{h} \sin(\omega_ht+\phi)\big] \\
   \hat{\theta} & = {1\over2} \arctan \bigg\{ {Y_\beta \over Y_\alpha - L_0} \bigg\}
\end{aligned}
\endequ
with $\phi=0$ and
$$
Y:=\col(Y_\alpha, Y_\beta),
$$
guarantees
$$
\limsup_{t\to\infty}\big|\Hat{\theta}(t) - \theta(t)\big| = n\pi + \mathcal{O}(\varepsilon^{1\over2})
$$
for $n\in\mathbb{Z}$, when $\omega_h \ge \omega_h^\star$ for some $\omega_h^\star >0$, with $\varepsilon= {2\pi \over \omega_h}$.
\end{proposition}

{Before presenting the proof, let us say a few words regarding the proposition intuitively. The widely popular LTI filtering technique for signal injection, illustrated in Fig. \ref{fig:block1}, is compactly expressed in \eqref{sig_proc}. To analyze it quantitatively, we make a mild assumption of bounded angular velocity and its time derivative. Proposition \ref{prop-conventional} figures out the steady-state accuracy $\mathcal{O}(\varepsilon^{\hal})$ of the conventional LTI filtering method with the suggested parameters.}

\begin{proof}
Applying the operator $\hpf$ to \eqref{ident1}, we have\footnote{\rm We omit the exponentially decaying term $\et$ of filtered signals in the following analysis.}
\begequ
\label{hpf1}
\begin{aligned}
   \hpf[\iab] & = \hpf[\biab] + {2\pi\over\omega_h}\hpf[D(\theta)S] + \mathcal{O}(\varepsilon^2) + \et,
\end{aligned}
\endequ
with the definition
$$
D(\theta) := L^{-1}(\theta)\begmat{1\\0}.
$$
For the first term of \eqref{hpf1}, we have
$$
\begin{aligned}
 \hpf[\biab]
  = &  {2 \over (w_h+p)^2}[r_1(t)] \\
 r_1(t)  := &
    {\partial \calf \over \partial \biab}\cdot (\calf+ L^{-1}\vab^C)
   +\Big( {\partial \calf \over \partial \bar{\omega}}+ L^{-1}\Big)
  \mathcal{O}(\ell_\omega)
  \\
  & +  n_p\bar{\omega}
   \Big({\partial \calf\over \partial\bar{\theta}} +
  {\partial L^{-1} \over \partial \bar\theta} \vab^C \Big)   \\
\end{aligned}
$$
where we have used the assumption
$$
\big|(\bar\omega, \dot{\bar\omega}, \dot{v}_{\alpha\beta}^C) \big|
\le \ell_\omega
$$
in the last term, with
$$
\calf(\iab,\theta,\omega) := L^{-1}(\theta)
F(\iab,\theta,\omega).
$$
 There always exists a constant $\omega_h^\star \in \rea_+$ such that for $\omega_h > \omega_h^\star$
$$
 \hpf[\biab]  = {2 \over \omega_h^2}\cdot {\omega_h^2 \over (\omega_h + p)^2} \big[\mathcal{O}(1)
   \big].
$$
Some basic linear system analysis shows
$$
\left| {\omega_h^2 \over (\omega_h + p)^2} \big[\mathcal{O}(1)] \right| = \mathcal{O}(1),
$$
 thus yielding
$$
 \hpf[\biab] = \mathcal{O}(\varepsilon^2).
$$

For the second term in the right hand side of \eqref{hpf1}, we have
$$
\begin{aligned}
   {2\pi\over\omega_h}\hpf[D(\theta)S]
      =  &- {2V_h \over \omega_h} {1 \over (\omega_h + p)^2}[r_2(t)]
     \\
     r_2(t):= &
   a_1 \cos(\omega_ht) + a_2 \omega_h \sin(\omega_ht)
   \\ &
   - a_3 \omega_h\sin(\omega_ht)
   - \omega_h^2 D(\theta)\cos(\omega_ht).
\end{aligned}
$$
with
$$
\begin{aligned}
a_1  :={d\over dt}(n_p\omega D'(\theta)) , \;
a_2  :=n_p\omega D'(\theta), \;
a_3  :=\omega_h D'(\theta)n_p\omega,
\end{aligned}
$$
the derivatives of which are all bounded. If the parameter $\omega_h$ is large enough, we have
$$
 {2\pi\over\omega_h}\hpf[D(\theta)S] = {1 \over \omega_h}V_h D(\theta)\sin(\omega_ht) +  \mathcal{O}(\varepsilon^2).
$$

Therefore, the currents filtered by the HPFs become
\begequ
\label{filter_hpf}
\begin{aligned}
    y_{h}& := \hpf[\iab] \\
    &
    ={1 \over \omega_h}V_h D(\theta)\sin(\omega_ht) + \mathcal{O}(\varepsilon^2).
\end{aligned}
\endequ
Multiplying $\sin(\omega_ht + \phi)$ on both sides with $\phi=0$, we get
\begequ
\label{filter_sin}
    \sin(\omega_ht)y_h = {V_h \over 2\omega_h}D(\theta) - {V_h \over 2\omega_h} D(\theta) \cos(2\omega_h t) + \mathcal{O}(\varepsilon^2),
\endequ
where we have used the trigonometric identity
$$
\sin^2\theta = {1 \over 2} (1-\cos2\theta).
$$
Applying the LPF to \eqref{filter_sin}, for the first term we have
$$
\begin{aligned}
    \lpf\bigg[ {V_h \over 2\omega_h}D(\theta)\bigg]
    = {V_h \over 2\omega_h}D(\theta) +
    \mathcal{O}(\varepsilon^{3\over2}).
\end{aligned}
$$
For the second term, we have
$$
    \lpf\bigg[ {V_h \over 2\omega_h} D(\theta) \cos(2\omega_h t)\bigg] =
    \mathcal{O}(\varepsilon^{3\over2}),
$$
with straightforward calculations and the swapping lemma.

Therefore, the filtered signal satisfies
$$
    \begin{bmatrix} y_\alpha \\ y_\beta \end{bmatrix}  := \lpf[\sin(\omega_ht)y_h]
     = {V_h \over 2\omega_h}D(\theta) +
    \mathcal{O}(\varepsilon^{3\over2}).
$$
Notice the explicit form of $D(\theta)$, thus we having
$$
\begin{aligned}
\left[\begin{aligned}
 Y_\alpha\\
 Y_\beta
\end{aligned}\right] &:=
\left[\begin{aligned}
 {2\omega_hL_dL_q \over V_h}y_\alpha\\
 {2\omega_hL_dL_q \over V_h}y_\beta
\end{aligned}\right]\\
& \;
=
\left[\begin{aligned}
 L_0 - L_1\cos2\theta \\
 - L_1\sin 2\theta
\end{aligned}\right]
+\mathcal{O}(\varepsilon^{1\over2}).
\end{aligned}
$$

It is obvious that when time $t$ goes to infinity, the identity
$$
\tan2\theta
 = {Y_\beta \over Y_\alpha - L_0} +\mathcal{O}(\varepsilon^{1\over2})
$$
holds. Thus,
$$
\theta = {1 \over 2} \arctan \bigg\{ {Y_\beta \over Y_\alpha - L_0}\bigg\} + {n\pi} + \mathcal{O}(\varepsilon^{1\over2}) + \et,
$$
with $n\in\mathbb{Z}$. It completes the proof.
\end{proof}

The saliency-tracking-based method has an angular ambiguity of $\pi$. It is possible to utilize the saturation effect in $d$-axis of machine, as well as $y_v$ to conduct the magnetic polarity identification. The problem is out of scope of the paper, and we refer the readers to \cite{JANetaltia,JEOetal} for more details. {It should be underscored that the above quantitative analysis does not rely the constant speed assumption, making it also applicable to transient stages in low speed region.}

When the IPMSM is working at standstill, the estimation accuracy at the steady stage becomes also $\mathcal{O}(\varepsilon)$. A corollary at standstill is given as follows.
\begin{corollary}
\rm
For Proposition \eqref{prop-conventional} with $\omega \equiv 0$, we have
$$
\limsup_{t\to\infty}\big|\Hat{\theta}(t) - \theta(t)\big| = n\pi + \mathcal{O}(\varepsilon)
$$
with $n\in\mathbb{Z}$.
\end{corollary}
\begin{proof}
It follows clearly with $\lambda_\ell = 1$.
\end{proof}

{The above corollary underscores that the position estimation at standstill admits the same order of accuracy with the one of the proposed virtual output estimator. This is because the angular position $\theta$ degenerates as a constant for such a case.}

%
\section{\cfm{\large Proposed Estimation Method}}
\label{sec4}
%
%
\subsection{New Design}

In this section, we propose a new estimator following the methodology in \cite{YIetalcst}. Before presenting the new design, we define three ${\tt BIBO}$-stable, linear operators,
\begin{itemize}
  \item first, the delay operator $\cald_d$, with parameter $d>0$,
  \begequ
\lab{hd}
\mathcal{D}_{{d}}[u(t)]= u(t-d);
\endequ
    \item second, the weighted zero-order-hold operator $\calz_w$, parameterized by $w>0$, and defined as
\begequ
\lab{zw}
\begin{aligned}
\dot{\chi}(t) & = u(t) \\
\calz_w[u(t)] & ={1\over w} \big[ \chi(t) - \chi(t- w) \big];
\end{aligned}
\endequ
    \item third, the linear time-varying (LTV) operator $\hgrad$ defined as
\begequ
\label{hgrad}
\begin{aligned}
    \dot{x}(t) & = - \gamma S^2(t)x(t) +  \gamma S(t)u(t) \\
        \hgrad[u(t)] & = {1 \over \varepsilon } x(t),
\end{aligned}
\endequ
where $\gamma>0$ is a tuning gain.
\end{itemize}
These operators are instrumental for the following design.

To construct the estimator, apply the first two operators to the currents as follows,
\begequ
\lab{Y}
Y_f(t): = (\mathcal{D}_d  - \mathcal{Z}_{2d})[i_{\alpha\beta}(t)].
\endequ
We make the observation that, using the Laplace transform, the action of \eqref{Y} may be represented in the frequency domain as
$$
Y_f(s)=G_d(s) \iab(s),
$$
where we defined the transfer function
\begequ
\lab{gd}
 G_d(s) := e^{-ds} +\dfrac{1}{2ds}\left( e^{-2ds} -  1\right).
\endequ
The description of the estimator is completed applying the third operator $\hgrad$ to $Y_f$ to generate the estimate of $y_v$, denoted as $\hat y_v$, that is,
\begequ
\label{hatyv}
        \hat{y}_v(t) = \hgrad[Y_f(t)].
\endequ

See Fig. \ref{fig1}. {The measured stator current $\iab$ is first filtered by the transfer function $G_d(s)$ defined in \eqref{gd}, then going through a gradient descent operator $\mathcal{G}_{\tt grad}$, and we will get the estimate of the virtual output $y_v$. As discussed in Section \ref{sec:32}, it is equivalent to obtaining the angular position $\theta$, which is our target. This is the signal processing procedure of our new design, the properties of which will be introduced below.}

\begin{figure}[h]
    \centering
    \includegraphics[width=9cm]{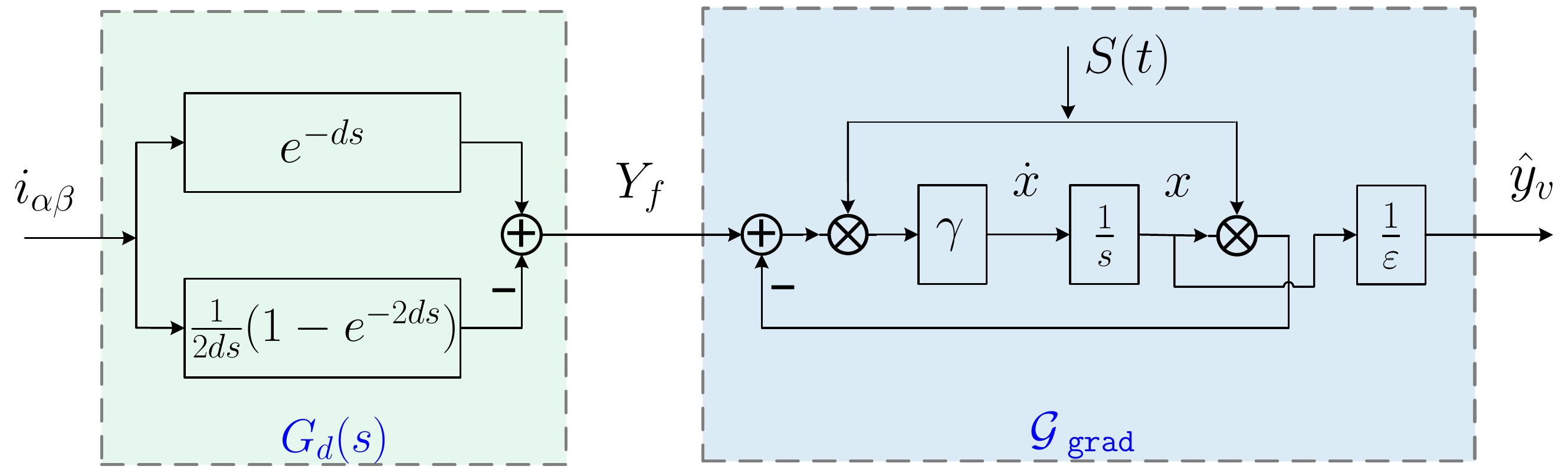}
    \caption{Block diagram of the proposed estimation method}
    \label{fig1}
\end{figure}

Using the analysis reported in \cite{YIetalcst}, with $d=\varepsilon$ it is shown that the estimator \eqref{Y}, \eqref{hatyv} verifies
$$
\lim_{t\to \infty} |\hat{y}_v(t) - y_v(t) | \le \mathcal{O}({\varepsilon}).
$$
A rigorous statement is given as follows.

\begin{proposition}
\label{prop-new-design}\rm
For the dynamical model of IPMSMs \eqref{sal_pmsm1}, suppose the control $\vab^C$ guarantees all states bounded and the speed
\begequ
\label{ass-dotyv}
|\dot{y}_v| \le \ell_v
\endequ
for some constant $\ell_v$, there exist constants $\omega_h^\star, \gamma^\star >0$ such that for $\omega_h> \omega_h^\star$ and $\gamma > {\gamma^\star \over\varepsilon}$, the estimate satisfies
$$
\limsup_{t\to\infty} \Big|\hgrad^\gamma \circ G_d(s) [\iab(t)] - y_v(t)  \Big| = \mathcal{O}(\varepsilon)
$$
where $y_v$ is defined in \eqref{yv} with $d=\varepsilon$.
\end{proposition}
\begin{proof}
We give a brief outline here. With the Taylor expansion, we can obtain the time-varying regressor
\begequ
\label{prop2-eq1_A}
Y_f(t)  = S(t) \theta_v(t-d) + \mathcal{O}(\varepsilon^2),
\endequ
with
$$
Y_f(t): = G_d(s)[\iab(t)],\quad \theta_v(t) = \varepsilon y_v(t).
$$
Define the error signal
$$
\Tilde{\theta}_v:= \Hat{\theta}_v - \varepsilon y_v .
$$
Invoking \eqref{prop2-eq1_A} and Assumption \eqref{ass-dotyv}, if $\gamma>{\gamma^\star \over \varepsilon}$ we get
\begin{equation}
\label{theta2dyn}
\dot{\tilde{\theta}}_v
    := - \gamma S^2(t)\tilde{\theta}_v +\mathcal{O}\big(\varepsilon\big).
\end{equation}

Clearly, the signal $S(t)$ is of persistent excitation, that is,
$$
    \int_{t}^{t+ {1 \over \varepsilon}}S^2(\tau) d\tau \ge S_0,
$$
for all $t \ge 0$ and some $S_0>0$. Invoking Krasovskii's theorem and carrying-out some basic perturbation analysis, it completes the proof.
\end{proof}

Thus, defining the angle estimate as
\begequ
\label{angle_est}
\hat{\theta} = {1 \over 2} \arctan \left\{ {\hat y_{v_2}\over \hat y_{v_1} - {L_0\over L_dL_q}}\right\},
\endequ
the required asymptotic accuracy \eqref{estacc} is achieved.

\vspace{0.15cm}

\noindent{\bf R5} {The above analysis shows the accuracy enhancement of the new design, compared with the conventional LTI filtering technique. It should be pointed out that all the operators used in the new design are \emph{linear}, making the algorithm highly efficient. Indeed, the proposed algorithm can be regarded as a linear time-varying (LTV) operator.}

\subsection{Parameter Tuning}

For the implementation of the proposed estimator, there are three tunable parameters, namely, $\gamma$, $\omega_h$ and $V_h$. We are in position to give some discussions on their selections.

\
\begin{itemize}
    \item[\bf{D1}] A larger gain $\gamma$ yields a faster convergence speed of the estimation error. The performance assessment, including the transient and steady-state stages, of the virtual output estimates is instrumental for its tuning.

    \vspace{0.3cm}

    \item[\bf{D2}] For the frequency parameter $\omega_h$, there is a tradeoff between the estimation accuracy and the sensitivity to unavoidable high-frequency measurement noises \cite{PETetal}. On one hand, a higher frequency increases the accuracy. On the other hand, the measurement is
    \begequ
\label{ident_noise}
\iab = \biab + \varepsilon y_v S  + \mathcal{O} (\varepsilon^2) + \nu,
    \endequ
where $\nu$ is the measurement noise. Thus, a higher frequency will decrease the signal-to-noise ratio.

    \vspace{0.3cm}

    \item[\bf{D3}] The amplitude $V_h$ shares some similar effects on the estimation performance with $\omega_h$, due to \eqref{ident_noise}, since, a larger amplitude $V_h$ increases the signal-to-noise ratio. {However, if the probing amplitude $V_h$ is sufficiently large, the oscillation in mechanical coordinates is not neglectable.}
\end{itemize}

\subsection{Model with Magnetic Saturation}
{To clarify the underlying mechanism of conventional LTI filtering methods we consider in this paper the basic IPMSM model \eqref{volt_model}. The proposed analysis may be adapted, almost {\em verbatim}, to more general models, {\em e.g.}, the model with magnetic saturation, which we proceed to discuss below.

Considering the flux linkage $\lambda\in \rea^2$ and Faraday's law, we have
\begequ
\label{eq:lambda}
\dot{\lambda} = \vab - R\iab,
\endequ
which holds for the IPMSM models. The stator current satisfies the constitutive relation
\begequ
\label{H_E}
\iab = \nabla_\lambda H_E(\lambda, \theta),
\endequ
where $H_E(\lambda,\theta)$ is the magnetic energy stored in the inductors. The two equations above hold true independently of the consideration of magnetic saturation, whose effect is captured in the energy function. The model of the {\em unsaturated} IPMSM model \eqref{volt_model} is completed replacing
\begequ
\label{eq:qua}
H_E(\lambda, \theta) = {1\over 2} [\lambda - \textbf{c}(\theta)\Phi]^\top  L^{-1}(\theta)[\lambda - \textbf{c}(\theta)\Phi]
\endequ
with $\textbf{c}(\theta):= \col(\cos\theta, \sin\theta)$, with \eqref{eq:lambda}-\eqref{H_E}. As pointed out in \cite{JEBetal} under high-load condition, magnetic saturation should be taken into account in the model. In such a case, the magnetic energy, is more complicated than the quadratic form \eqref{eq:qua}, see \cite{BASetal,JEBetal,LI,BIA,SER,GUGetal} for more details.

The analysis of virtual outputs extraction in Sections \ref{sec3}-\ref{sec4} is \emph{exactly the same} for \eqref{eq:qua} and {\em non-quadratic} magnetic energy $H_E(\lambda,\theta)$. The difference appears on the way how to recover the angular position $\theta$ from $y_v$. For the latter, it can be formulated as the following optimization problem
$$
\hat{\theta} = \min_{\hat\theta \in [0,2\pi)} \left| \hat{y}_v - \nabla^2_\lambda H_E(\lambda,\theta)
\begin{pmatrix} 1 \\0  \end{pmatrix}
  \right|,
$$
whose analytical solution, if it exists, depends on the particular modeling of magnetic saturation. See \cite{JEBetal} for the discussions about saturated models.
}
%
\section{\cfm{\large Frequency Interpretation of New Estimator}}
\label{sec5}
%
We have shown that the new estimator \eqref{estacc} achieves performance enhancement with a higher accuracy. In this section we show the structural and functional similarities between two methods.

The new estimator proposed in this paper \emph{exactly} coincides with the block diagram in Fig. \ref{fig:block1} assigning $\phi={3\pi\over 2}$ and the filters as follows
\begali{
\nonumber
{\tt HPF} &= \mathcal{D}_d - \mathcal{Z}_{2d}\\
\lab{filters2}
{\tt LPF} &= {1\over 2} {\bigg({V_h\over 2\pi}\bigg)}^2 \calh,
}
where $\calh$ is the single-input single-output LTV filter
\begequ
\lab{grad}
\begin{aligned}
\dot{z}(t) & = -\gamma S^2(t)z(t)+ \gamma  u(t)\\
\calh[u(t)] & =z(t).
\end{aligned}
\endequ
See  Fig. \ref{fig3}.

\begin{figure}[h]
    \centering
    \includegraphics[width=9cm]{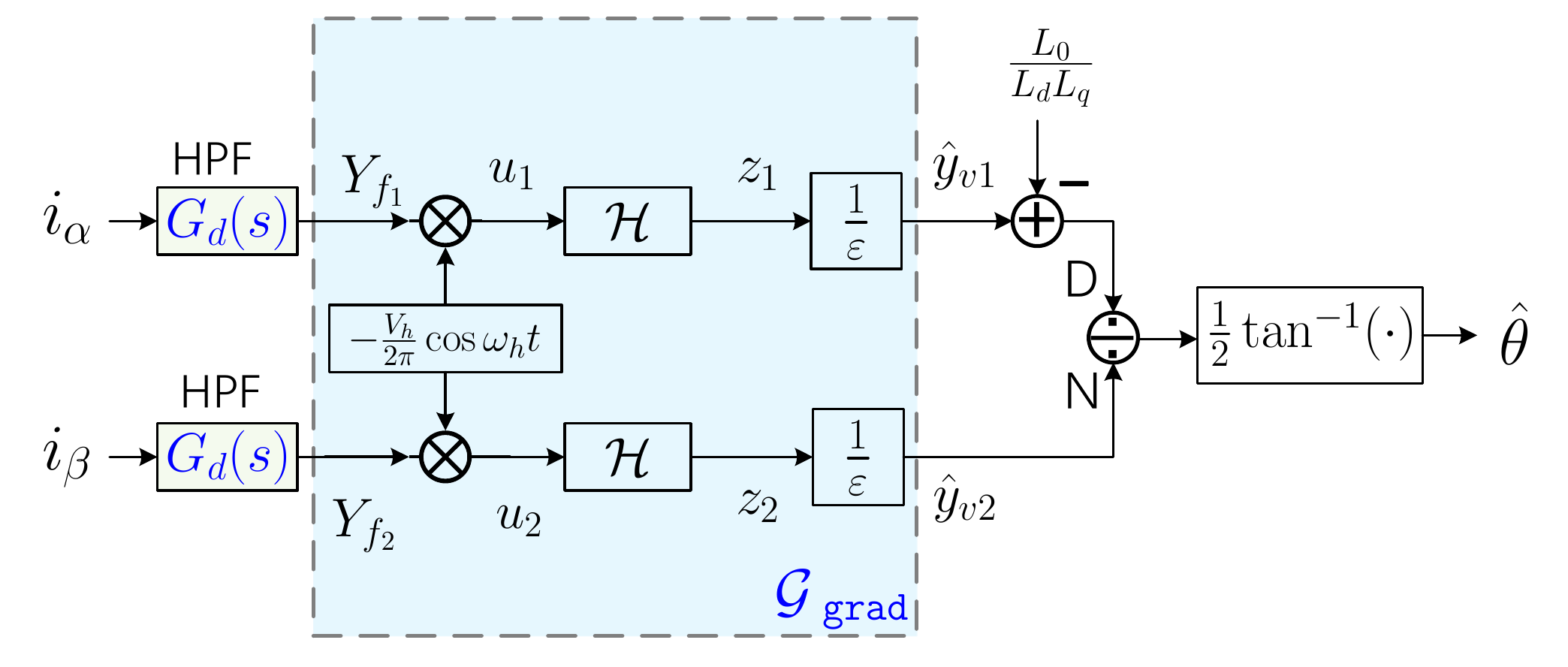}
    \caption{Equivalent block diagram of the proposed estimation method}
    \label{fig3}
\end{figure}

To illustrate the high-pass and low-pass filtering properties of \eqref{filters2}, \eqref{grad} we show in Fig. \ref{fig:bode_HPF2} the  Bode diagram of the transfer function $G_d(s)$, defined in \eqref{gd}. From the figure it is clear that $G_d(s)$ verifies a high-pass property.
\begin{figure}[h]
    \centering
    \includegraphics[width=5cm,height=4.5cm]{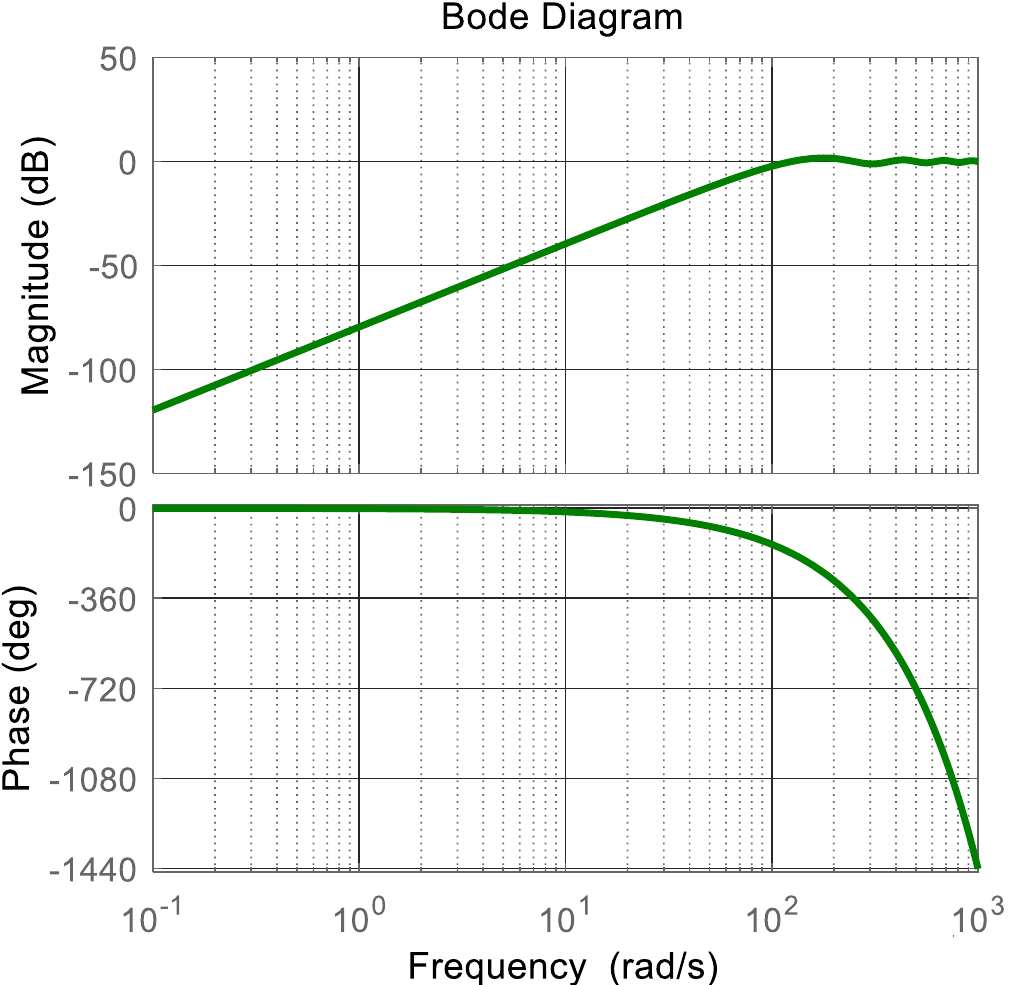}
    \caption{Bode diagram of the transfer function $G_d(s)$ ($\omega_h=500,\; n=2$)}
    \label{fig:bode_HPF2}
\end{figure}

The frequency response of the operator $\mathcal{H}$, by fixing $\gamma=1$ {and $V_h=1$}, is the same as the linear time periodic (LTP) system below
\begequ
\label{ltpsyst}
 \dot{y} = - (\cos\omega_ht)^2 y +u.
\endequ

Introducing the change of state coordinate
$$
y(t)= {\bar{P}^{-1}(t)}r(t)
$$
with
\begequ
\label{funp}
\bar{P}(t)  := \exp\left(-{1\over 4\omega_h}\sin2\omega_ht \right),
\endequ
we get the LTV system
\begequ
\label{LTV}
\begin{aligned}
    \dot{r}(t)  & = -{1\over2}r(t) +  \bar{P}(t)u(t).
\end{aligned}
\endequ

Therefore, the LTP system \eqref{ltpsyst} can be represented as in Fig. \ref{fig:LTP}. From \eqref{funp} it is clear that the matrix $\bar{P}(t)$ is almost identity for sufficiently large $\omega_h$, while the transfer function ${1\over s +0.5}$ admits the low-pass property for a large $\omega_h$. Hence, the LTP system \eqref{ltpsyst}, as well as the operator $\mathcal{H}$, is an LPF.
\begin{figure}[h]
    \centering
    \includegraphics[width=5cm]{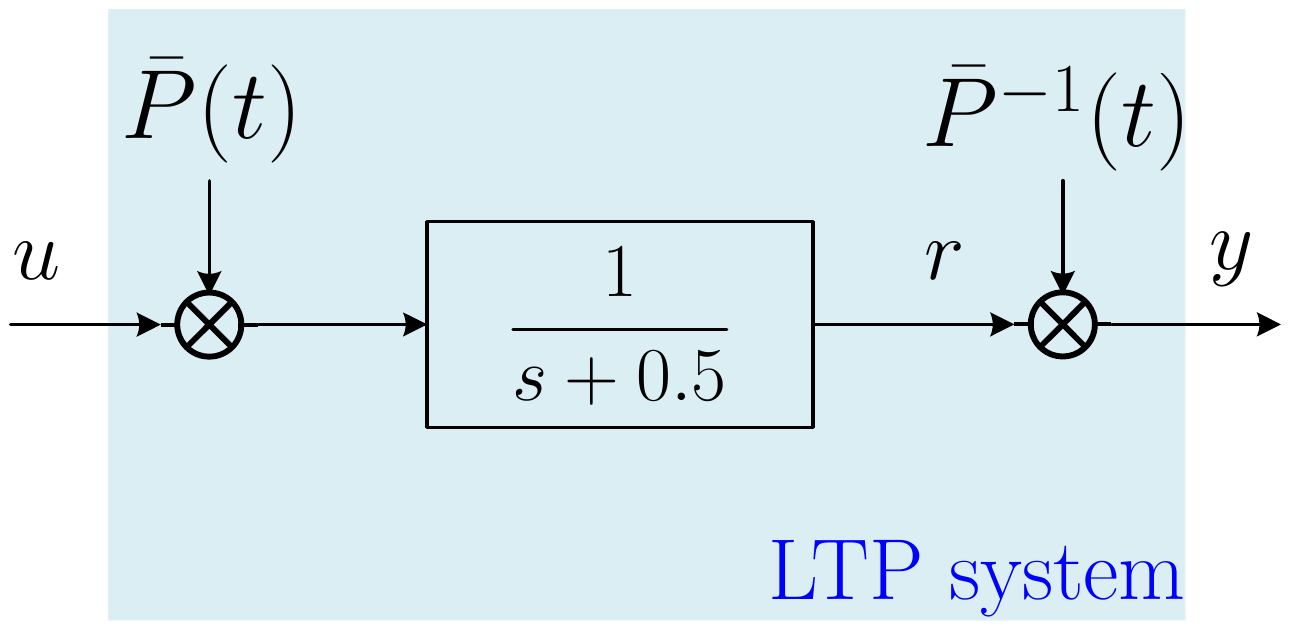}
    \caption{Equivalent block diagram of the LTP system \eqref{ltpsyst}.}
    \label{fig:LTP}
\end{figure}
%
In Fig. \ref{fig:block1}, the signal $\sin(\omega_ht+ \phi)$, entering before the LPF, has different parameters $\phi$ for the classical design and the proposed one, which are $0$ and ${3\pi\over2}$, respectively. As shown in Figs. \ref{fig:bode_H/LPF} and \ref{fig:bode_HPF2}, this is caused by the different phase lags ($+{\pi \over 2}$ and $-2\pi$) of the second-order LTI filter \eqref{filters1} and the delayed LTI filter $G_d(s)$ at the frequency $\omega_h$.
%

%
\section{\cfm{\large Simulations and Experiments}}
\label{sec6}
%
\subsection{Simulations}
The proposed estimator is first tested by means of simulations in Matlab/Simulink. We use the parameters of Table \ref{tab:1}, the current-feedback controller $\Sigma_C$ given below, together with the proposed estimator.

\begin{enumerate}
  \item[1)] Position estimator in Fig. \ref{fig1} with \eqref{angle_est}.
  \item[2)] Rotation between $\alpha\beta$-coordinates and misaligned $dq$-coordinates, namely,
    $$
  i_{dq} = e^{-\calj \hat\theta}\iab, \quad \vab = e^{\calj \hat\theta}v_{dq}.
  $$
  \item[3)] Speed regulation PI loops
  $$
  i_{dq}^\star = \left( K_p + K_i {1\over p}\right) (\omega^\star - \hat\omega),
  $$
  where $\omega^\star$ is the reference speed, and $\hat{\omega}$ is an estimate of the rotor speed obtained via the following PLL-type estimator.
\begali{
\nonumber
\dot{\eta}_1 & = K_p (\hat{\theta} - \eta_1) + K_i\eta_2\\
\nonumber
\dot{\eta}_2 & = \hat{\theta} - \eta_1 \\
\nonumber
\hat{\omega}_p & = K_p(\hat{\theta} -\eta_1) + K_i \eta_2 \\
\label{PLL}
\hat{\omega} & = {1\over n_p} \hat{\omega}_p.
}
  \item[4)] Current regulation loops
  $$
  \begin{aligned}
    v_d & = \left( K_p + K_i {1\over p}\right)(i_d^\star - i_d^\ell) - L{n_p}\hat\omega i_q \\
    v_q & = \left( K_p + K_i {1\over p}\right)(i_q^\star - i_q^\ell) + L{n_p}\hat\omega i_d + {n_p}\hat\omega \Phi,
  \end{aligned}
  $$
  where $i_{dq}^\ell$ are filtered signals of $i_{dq}$ by some LPFs.
  \end{enumerate}

\begin{table}[h]
\centering
\caption{Parameters of the IPMSM: {Simulation (First Column) and \\ Experiments (Second Column)}}
\label{tab:1}
\renewcommand\arraystretch{1.6}
\begin{tabular}{l|r|r}
\hline\hline
 Number of {pole pairs} ($n_p$) &  6 & 3  \\
 PM flux linkage constant ($\Phi$) [Wb] &0.11& 0.39\\
 $d$-axis inductance ($L_d$) [mH] & 5.74 & 3.38 \\
 $q$-axis inductance ($L_q$) [mH] & 8.68 & 5.07\\
 Stator resistance ($R_s$) [$\Omega$]& 0.43 & 0.47\\
 Drive inertia ($J$) [kg$\cdot\text{m}^2$] & 0.01 & $\ge $ 0.01\\
 \hline\hline
\end{tabular}
\end{table}

We operate the motor at the slow speed of 30 rad/min with $T_L=0.5$ N$\cdot$m and the parameters {$\varepsilon=10^{-3},\; \gamma=10^{4},\; V_h=1,\; \omega^\star=0.5$ and those in Table. \ref{tab:2}.
 \begin{table}[h]
\centering
\caption{Parameters of the controller and the PLL estimator}
\label{tab:2}
\renewcommand\arraystretch{1.6}
\begin{tabular}{l|r}
\hline\hline
 $[K_p,K_i]$ in the speed loop & $[1,5]$\\
  \hline
 $[K_p,K_i]$ in the current loop & $[5,5]$\\
  \hline
 $[K_p,K_i]$ in the PLL estimator & $[5,0.01]$\\
 \hline\hline
\end{tabular}
\end{table}
Fig. \ref{fig:simulation} shows the simulation results. In Fig. \ref{fig:siml1}, we also give the position estimate obtained from the conventional LTI filters, denoted $\hat{\theta}_{\mathtt{LTI}}$. {Considering the root-mean-square deviation (RMSD)
$$\mathtt{RMSD} = \sqrt{ {1\over t_2-t_1}\int_{t_1}^{t_2} |\hat{\theta}(s) - \theta(s)|^2 ds }
$$
with $\theta,\hat{\theta} \in \mathcal{S}^1$, we calculate the RMSDs for two methods in the interval $[5,10]$ s. They are  0.0872 and 0.1411 for the proposed design and the conventional LTI filtering method, respectively.} We conclude that the new design outperforms the conventional LTI filtering method with a higher accuracy. It is also observed that the sensorless control law regulates the angular velocity at the desired value. In Figs. \ref{fig:siml5}-\ref{fig:siml6}, we also test the position estimator, as well as closed-loop performance, with different injection frequencies, 500 Hz and 600 Hz, verifying the claim in \textbf{D2}.
\begin{figure}[ht]
\centering
\subfigure[The angle $\theta$ and its estimates (test frequency 1000 Hz)]{
    \includegraphics[width=0.22\textwidth,height=0.15\textwidth]{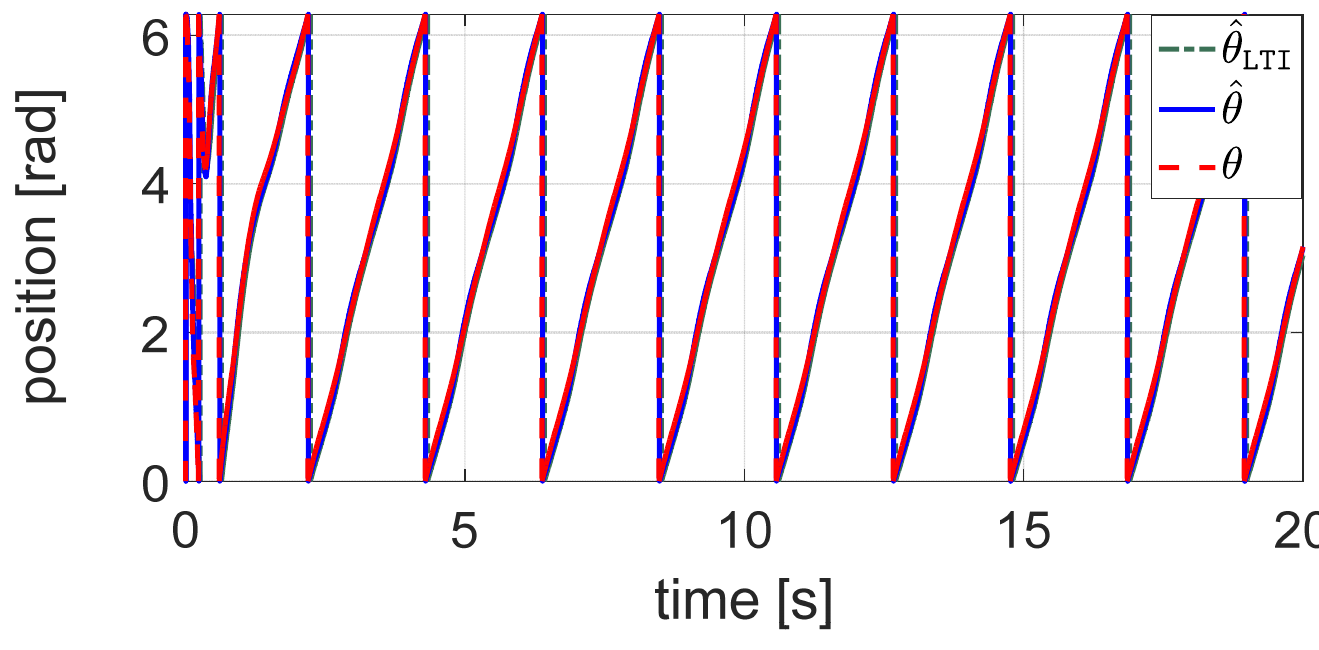}
    \label{fig:siml1}
}
\subfigure[{The angle $\theta$ and its estimates (test frequency 1000Hz)}]{
    \includegraphics[width=0.22\textwidth,height=0.15\textwidth]{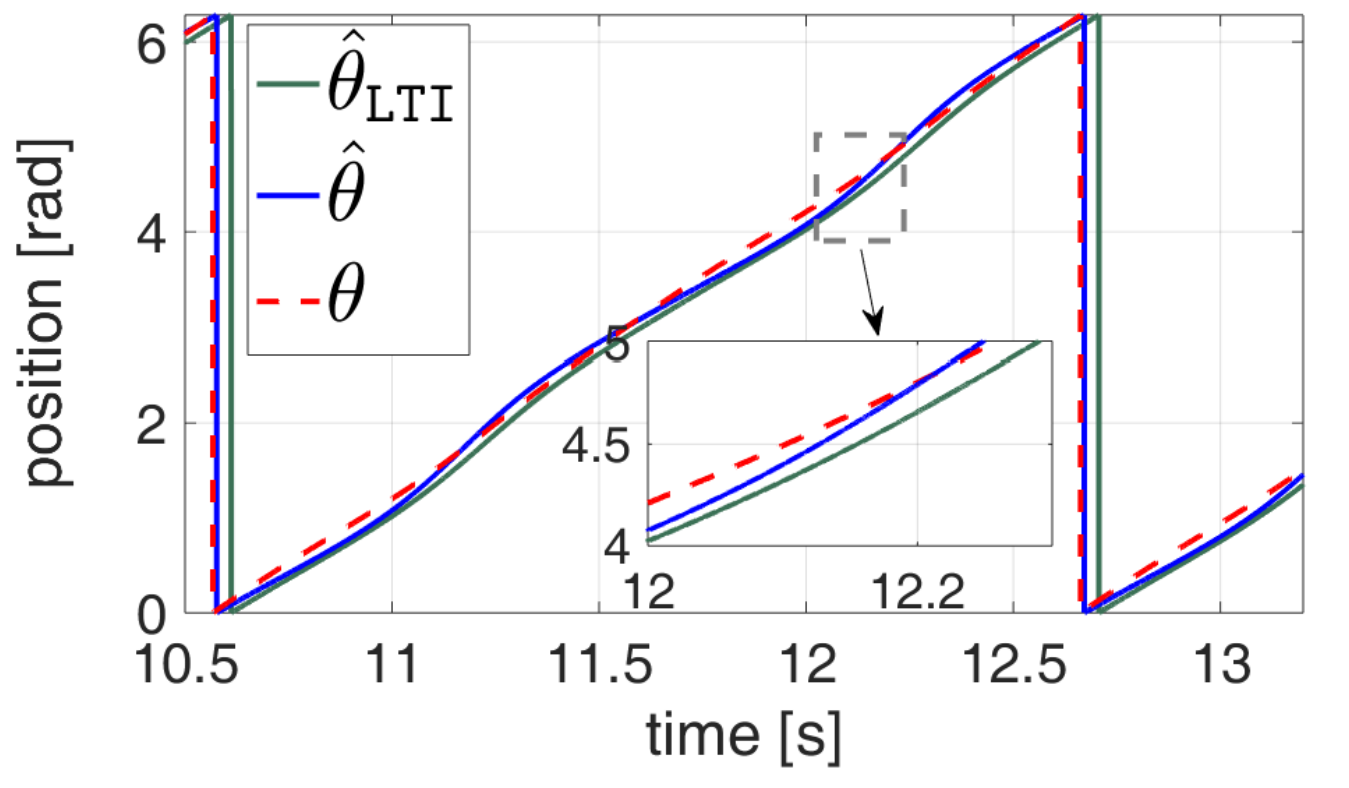}
    \label{fig:siml1_2}
}
\subfigure[Angular velocity and its estimate]{
    \includegraphics[width=0.22\textwidth,height=0.15\textwidth]{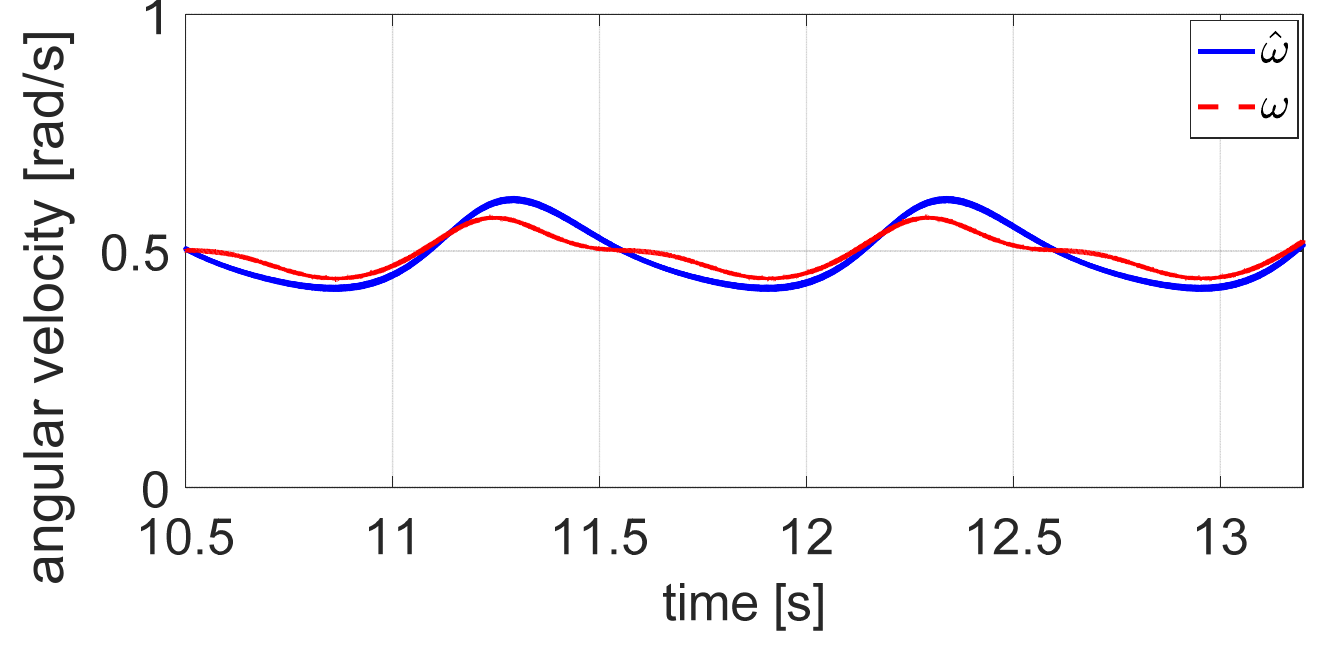}
    \label{fig:siml2}
}
\subfigure[Stator currents $\iab$]{
    \includegraphics[width=0.22\textwidth,height=0.15\textwidth]{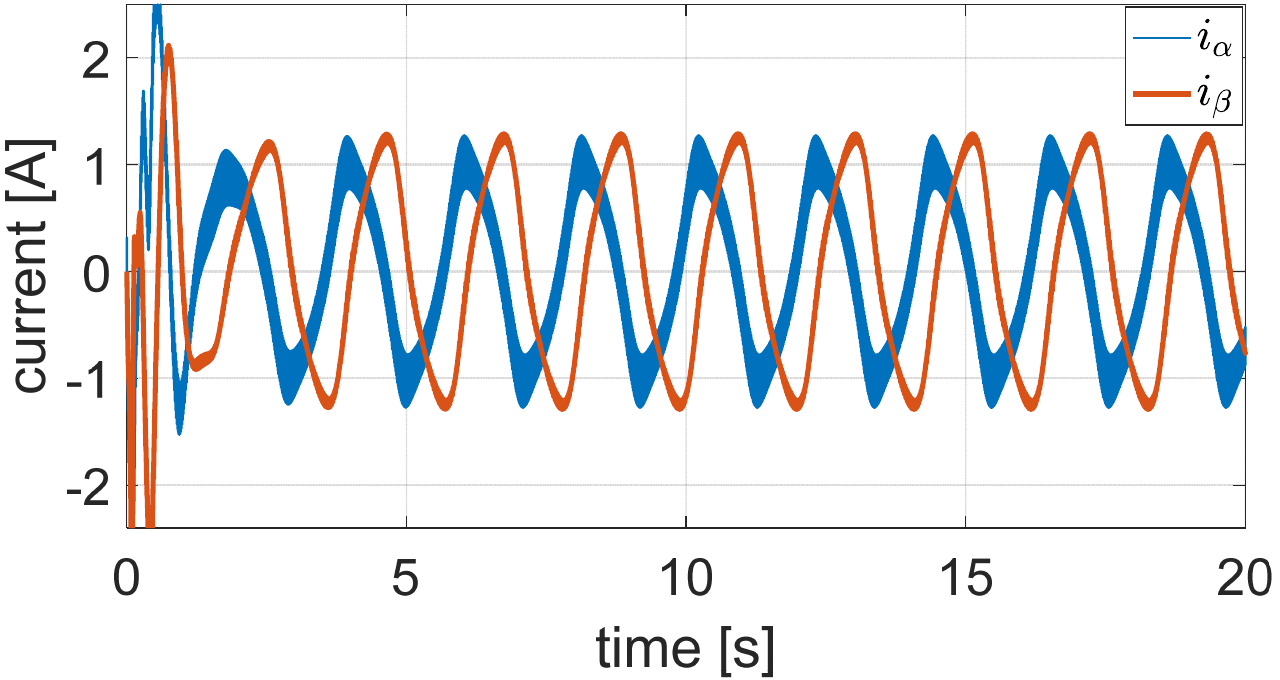}
    \label{fig:siml3}
}
\subfigure[The angle $\theta$ and its estimates (test frequency 600 Hz)]{
    \includegraphics[width=0.22\textwidth,height=0.15\textwidth]{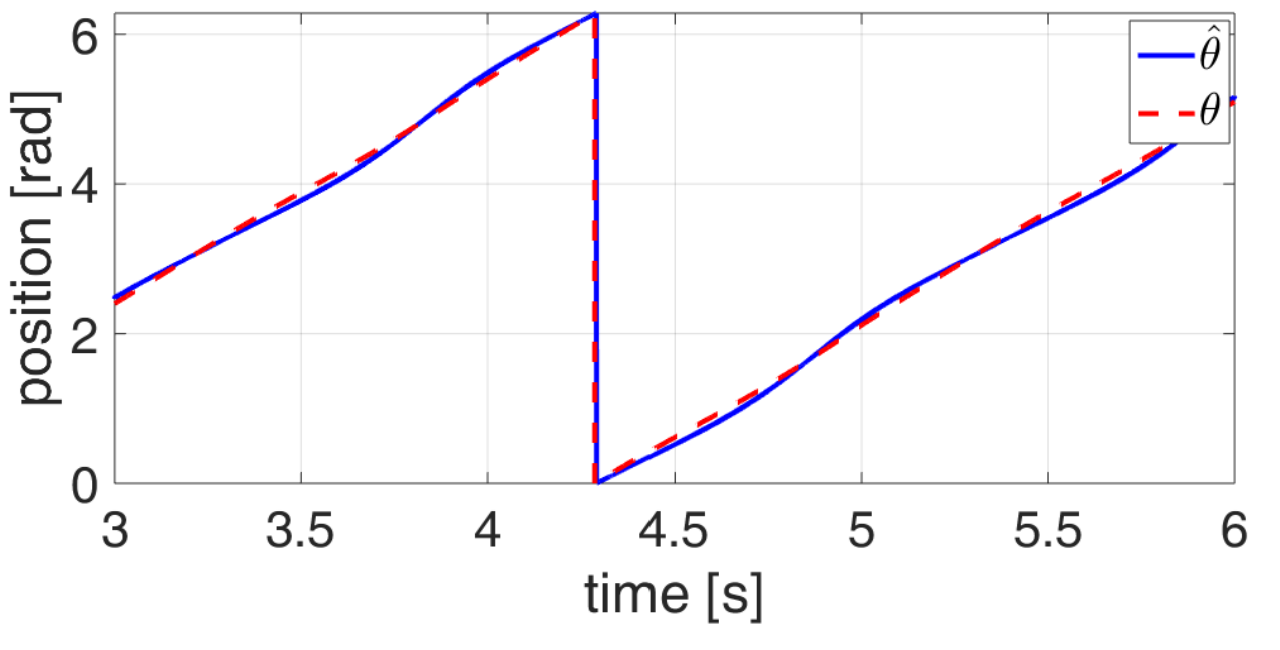}
    \label{fig:siml5}
}
\subfigure[The angle $\theta$ and its estimates (test frequency 500 Hz)]{
    \includegraphics[width=0.22\textwidth,height=0.15\textwidth]{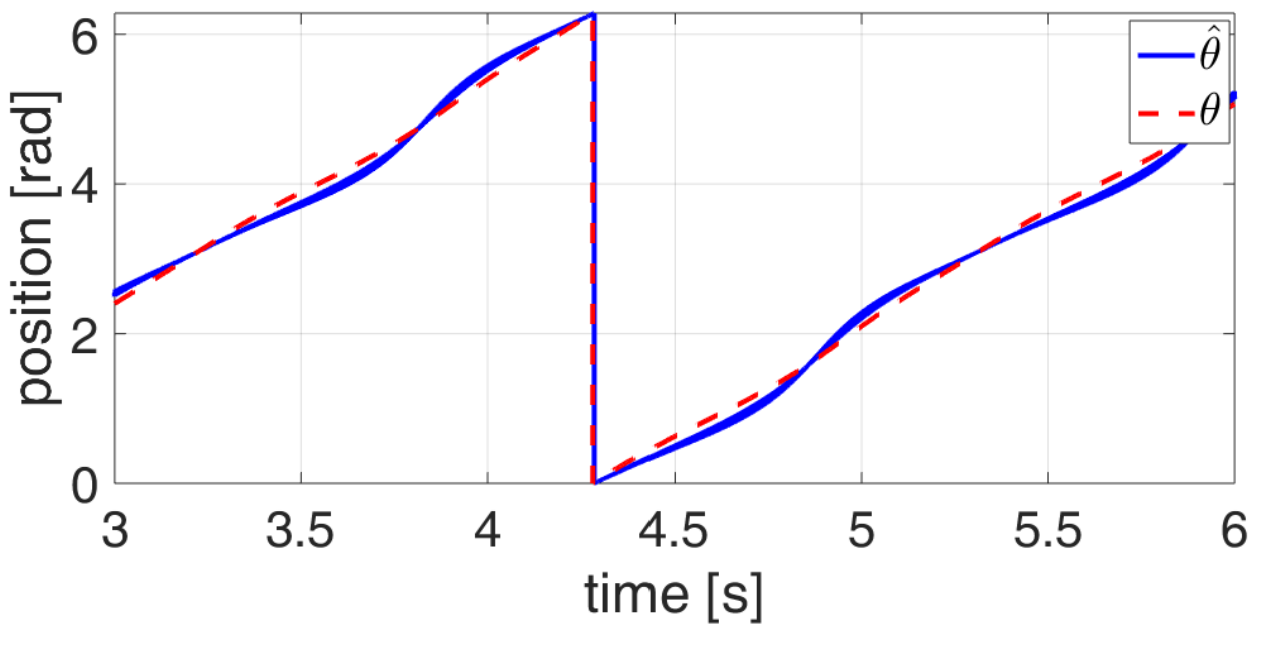}
    \label{fig:siml6}
}
\caption{Simulation results}
\label{fig:simulation}
\end{figure}

\subsection{Experiments}
\ \\ \
\textbf{Losses and Compensations.} Before introducing the experimental results, let us say something about the loss of phase shift, which is unavoidable, as well as its compensations.

The excitation signal is injected into the modulation signal for the stationary $\alpha$-axis. The excitation in the $\alpha$-axis is, indeed, also affected by the inverter imperfections, for instance, the lockout time. Further on, the current $i_\alpha$ responding to the excitation is phase-shifted by 90 degrees with respect to the voltage only in an ideal case, whereas the winding does not have any resistance. In practice, it should be a non-zero winding resistance, with the phase shift lower than 90 degrees. Additional impact is due to the excitation-frequency eddy-currents in the magnetic circuit. Acting as a short-circuited secondary winding in a transformer-like electromagnetic setup, the iron losses introduce an additional change of the phase shift. Thus, these factors, but not limited to them, cause the phase of high-frequency component in $\iab$, see the term $S(\cdot)$ in \eqref{ident1}, to be different from the one in the test signal.

Fig. \ref{fig:loss} illustrates the effect of the phase shift loss on the virtual output estimate $\hat{y}_{v_1}$ via simulations, namely, adding ``artificial'' phase shifts in the high-frequency components of \eqref{ident1} to study the changes of $\hat{y}_{v_1}$. In terms of \eqref{yv}, the virtual output $y_{v_1}$ admits the form
$
y_{v_1} = a \cos 2\theta + b
$
with some constants $a$ and $b$. In Fig. \ref{fig:loss} we observe the phase shift causes the drifts of the amplitude $a$ and the bias $b$. A natural compensation method is using the signal $S(t) = - { V_h } \cos (\omega_h t + \phi_p)$ rather than \eqref{S} in the gradient descent operator $\hgrad$, with $\phi_p$ tuned in $[0,2\pi)$.

Here we introduce an alternative approach to compensate all the losses, which not only contain the phase shifts but also the inductance values, at the ports of virtual output estimates. That is using the compensated virtual outputs
$$
\hat{y}_{v_1}^p := \ell_1 \hat{y}_{v_1} + \ell_2, \quad
\hat{y}_{v_2}^p := \ell_3 \hat{y}_{v_2},
$$
and corresponding angle estimate
$$
\hat{\theta} = {1 \over 2} \arctan \left\{ {\hat y_{v_2}^p\over \hat y_{v_1}^p - {L_0\over L_dL_q}}\right\}.
$$
The parameter adjustment principle of $\ell_i$ $(i=1,2,3)$ is to make the infimums and supremums of the signals $\hat{y}_{v_1}^p $ and $\hat{y}_{v_2}^p $ coincide with the ones calculated from \eqref{yv}. {For more details, we first run the motor at a constant speed, with the obtained estimates $\hat{y}_v$ approximating some periodic signals, the means and the amplitudes of which may differ from the calculations from \eqref{yv}. We then tune $\ell_i$ in order to let $\hat{y}_v$ have the same means and the amplitudes with the ones in \eqref{yv}. It is clear that $\ell_1$ and $\ell_3$ are proportional to the amplitudes, and $\ell_2$ affects the mean value. Such a procedure can be done offline, and then we apply the well-tuned parameters for online estimation.}

\ \\ \
\textbf{System Configuration.} The scheme developed in this paper was tested on an IPMSM platform, shown in Fig. \ref{fig:equip}. The test IPMSM is a FAST IPMSM, whose parameters are given in Table \ref{tab:1}. It has a 72 V line-to-line peak at 1000 RPM. The voltage of DC bus is 521 V, with the frequency of PWM 5 kHz.

The experimental setup comprises two synchronous motors, one of which runs in the speed control mode, used to maintain the speed at the desired level. The motors are coupled by means of a toothed belt, which also connects an inertial wheel. In Fig. \ref{fig:equip}, the experimental setup is also equipped with two mechanically coupled, inverter supplied brushless DC motors:  1) main power supply unit comprising the line rectifier and two 3-phase PWM inverters with control circuits, 2) DC-bus support with dynamic breaking, 3) speed controlled motor, 4) torque controlled motor, 5) inertia coupled with both motors. The motor under the test is obtained by taking an industry-standard FAST motor and introducing changes into the rotor magnetic circuit so as to obtain the difference (2:3) between the $d$-axis and $q$-axis inductance. This motor runs in the torque-control mode. The speed and position are obtained through the high-speed digital serial link from standard industrial high-resolution sensors mounted on the shaft. The sampling time is $T_s=300$ ns, and the acquisition time is set to cover at least two electrical periods. The three-phase currents, voltages and the rotor position were measured from the drive measurement system---a "Sincoder" shaft sensor.

\begin{figure}
\centering
\includegraphics[width=6cm,height=3cm]{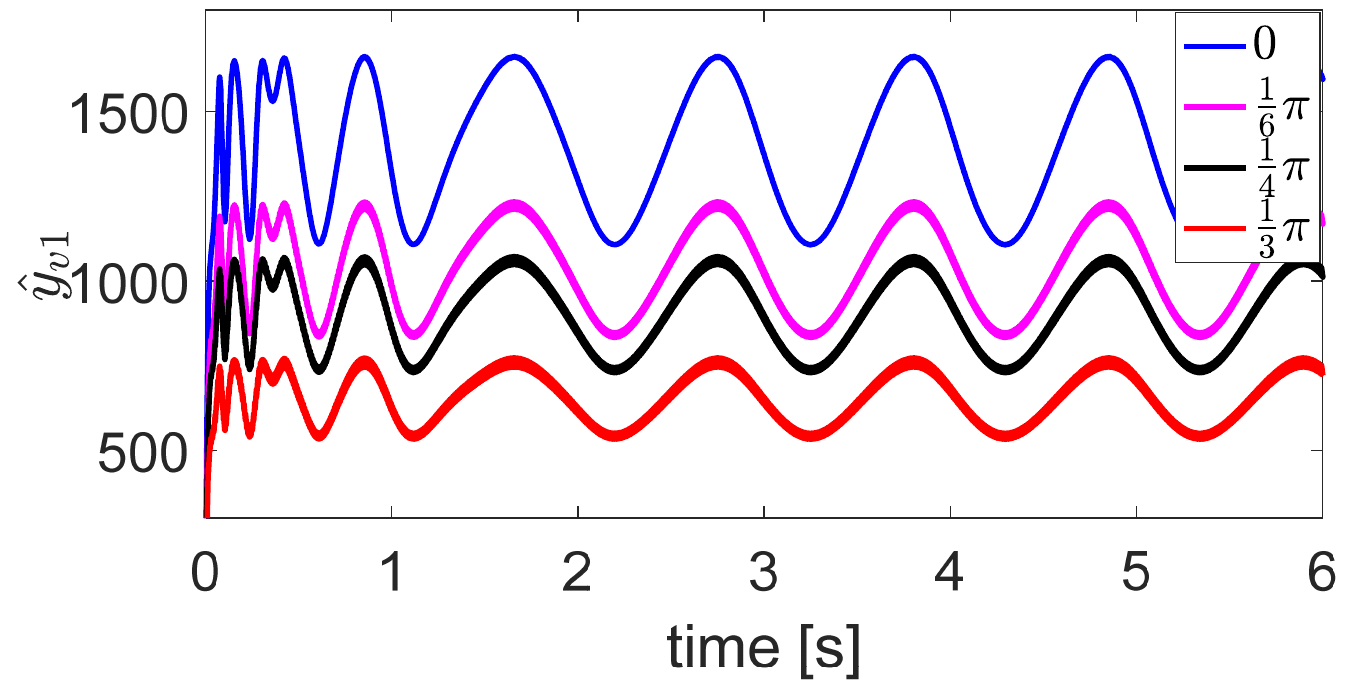}
\caption{The loss effect of shift drifts}
\label{fig:loss}
\end{figure}

\begin{figure}[ht]
\centering
    \includegraphics[width=5cm]{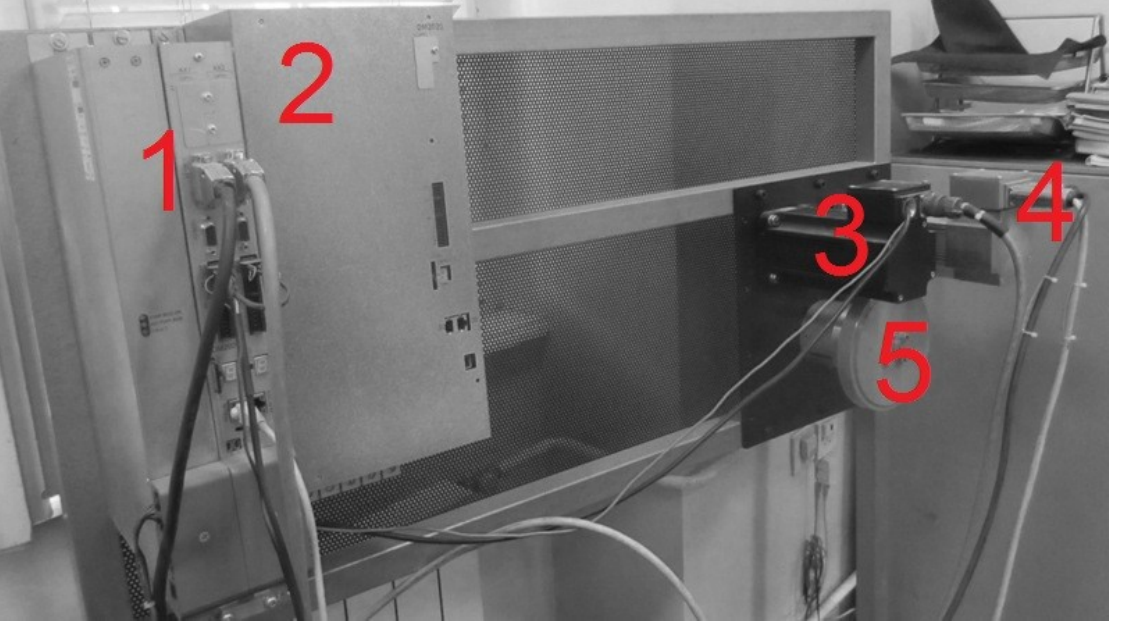}
\caption{Experimental testing setup}
\label{fig:equip}
\end{figure}


In experiments, we only test the proposed estimator, thus the estimated signals are not used in the closed-loop system, whose bandwidth of the speed loop is larger than 100 Hz. The test signal is injected into the modulating signal in the $\alpha$-axis. The test motor is driven by another motor which kept the speed at 60 RPM. 

\ \\ \
\textbf{Experimental Results.}
~In the first experiment, the amplitude of the test signal is 2 V, with the frequency 400 Hz {with zero reference currents ($i_d^d$ and $i_q^d$)}, in this way we can increase the resolution of the measurement system. The parameters of the estimator are selected as $\gamma_\alpha =1.25\times 10^4$ and $\gamma_\beta =2.5\times 10^4$. {Let us first consider a normal operation mode with a 400 Hz test signal and zero load, and run the motor at a constant speed. For such a case, Fig. \ref{fig:exp1} shows the performance of the proposed position estimator, which generally works well.} The test signal was only injected to the $\alpha$ axis, which is illustrated in Fig. \ref{fig:exp2} after Clarke transform. {Under the same conditions, we compare the proposed angular estimator with the conventional LTI filter in Fig. \ref{fig:comp1}, in which we observe the accuracy enhancement of the new design. This is probably because the phase lag in the virtual output estimates of the proposed design is smaller than that of the LTI filter, see Fig. \ref{fig:comp2}. Of course, the performance of the conventional methods may be further improved via some particular technique-oriented tricks in the applied literature, which, however, are out of scope of this paper.}

{The motor platform utilized 5 kHz PWM, thus we should select the frequency of the test signal less than one tenth of this value, {\em i.e.}, 500 Hz, in order to be able to ``generate" the desired sinusoidal signals. We first test an 800 Hz probing signal, but unfortunately, it does not work as expected.} We present some experimental results in Fig. \ref{fig:frequency} to illustrate estimation performance with the injection frequencies 200 Hz and 100 Hz, respectively. We conclude that the performances degenerate with the frequency varying from 400 Hz to 100 Hz. A possible interpretation, as discussed in \textbf{D2}, is that such a case has a low signal-to-noise ratio with a large $\omega_h$. {A similar result was observed for magnetic levitation systems in \cite{YIetalcst}.} Therefore, the injection frequency is suggested to be selected in [200, 400] Hz for the tested motor {in our experimental testing setup. For any other equipment, the probing frequency can be selected following the same procedure above, considering both the frequency of PWM signals and the signal-to-noise ratio}.

\begin{figure}[h]
\centering
\subfigure[Angle $\theta$ and its estimates]{
    \includegraphics[width=0.22\textwidth,height=0.15\textwidth]{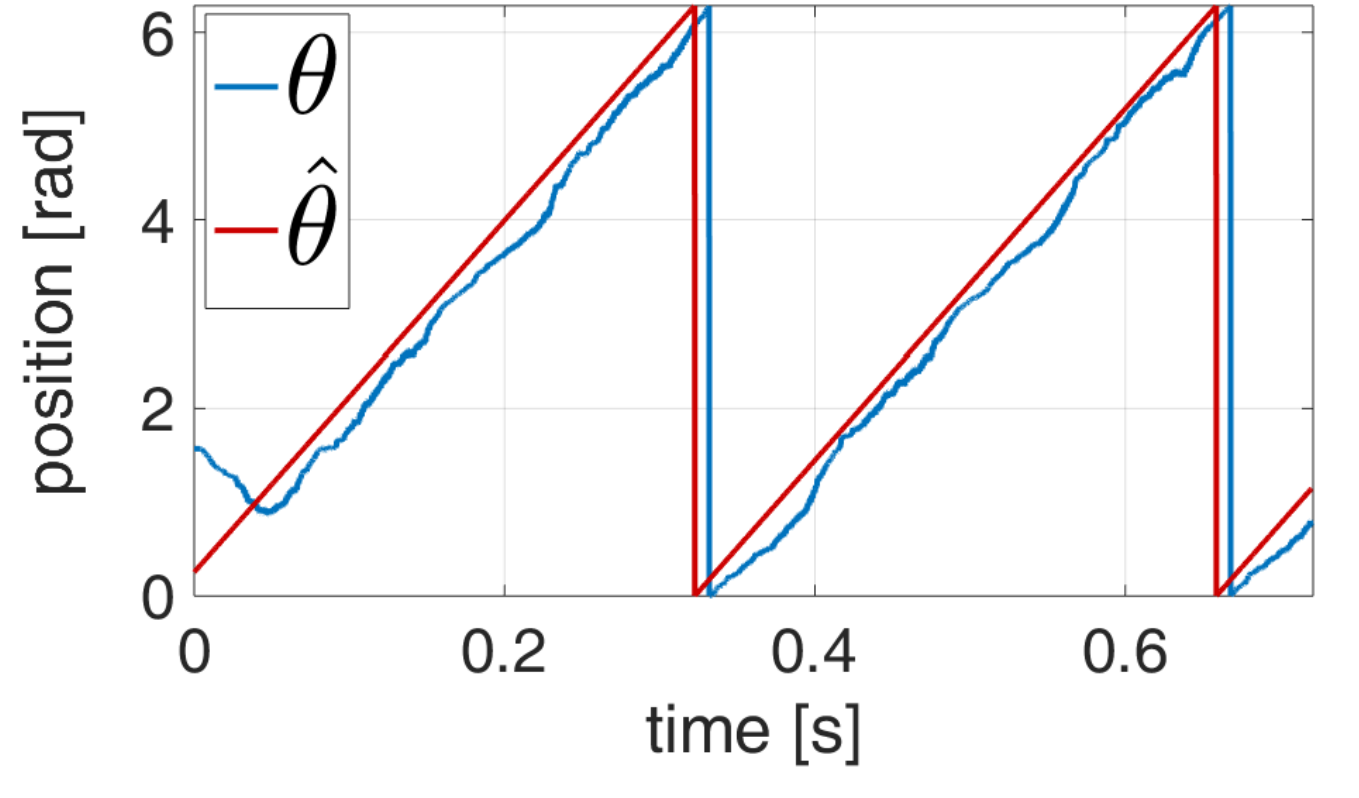}
    \label{fig:exp1}
}
\subfigure[Stator voltage $u_\alpha$]{
    \includegraphics[width=0.22\textwidth,height=0.15\textwidth]{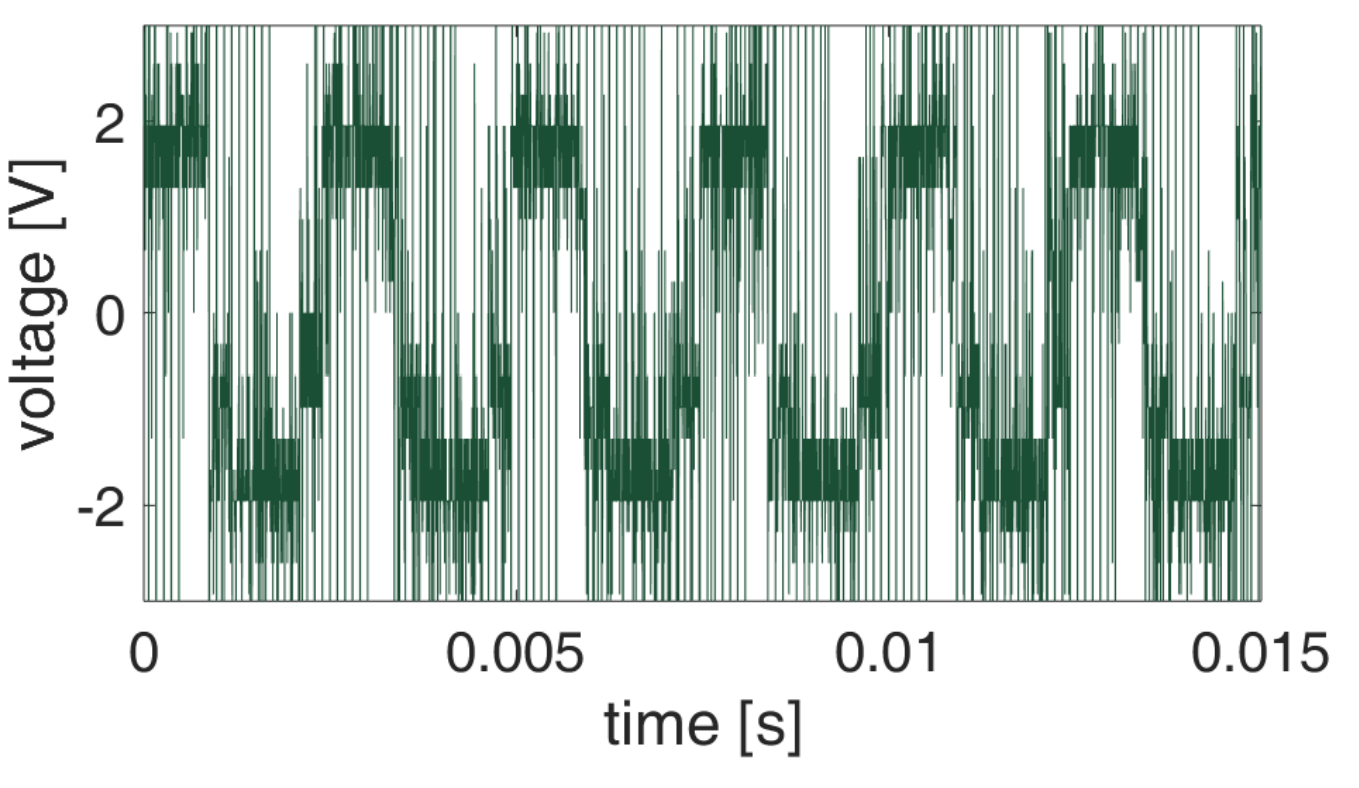}
    \label{fig:exp2}
}
\caption{Experimental results (zero load, 400 Hz injection signal, constant speed)}
\end{figure}

\begin{figure}[h]
\centering
\subfigure[Angle $\theta$ and its estimates]{
    \includegraphics[width=0.46\textwidth,height=0.23\textwidth]{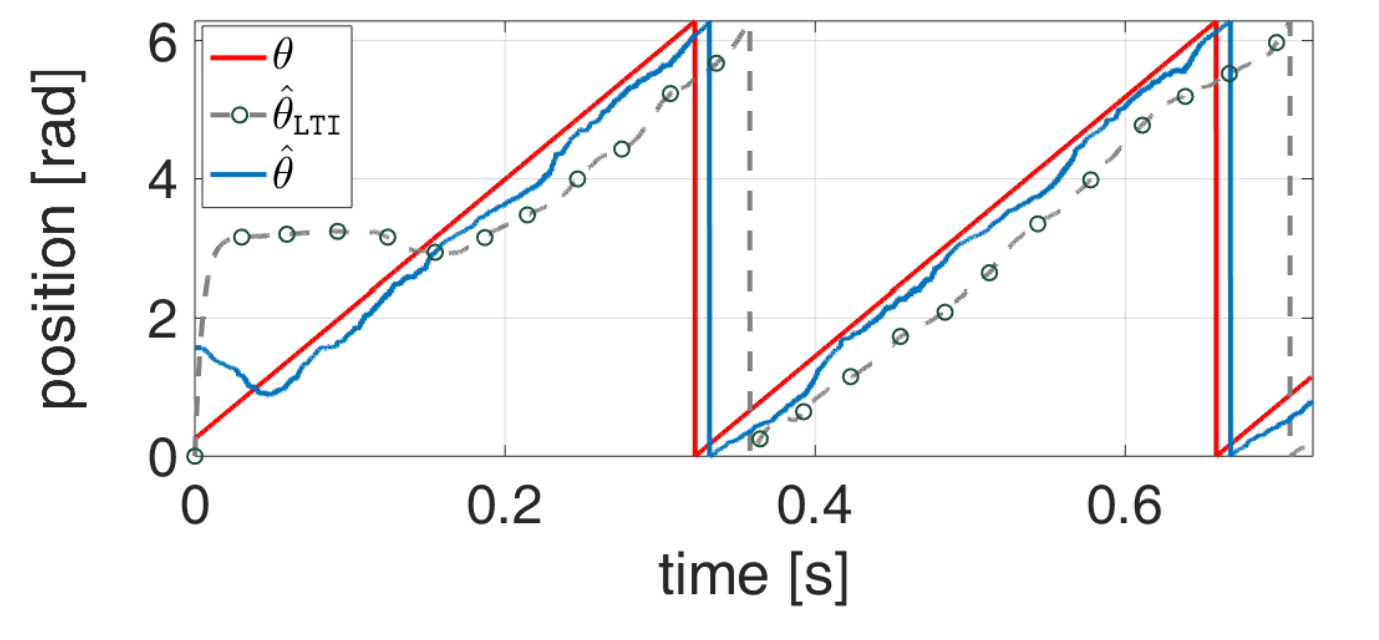}
    \label{fig:comp1}
}
\subfigure[Virtual output $y_{v1}$ and its estimates $\hat{y}_{v1}$]{
    \includegraphics[width=0.46\textwidth,height=0.23\textwidth]{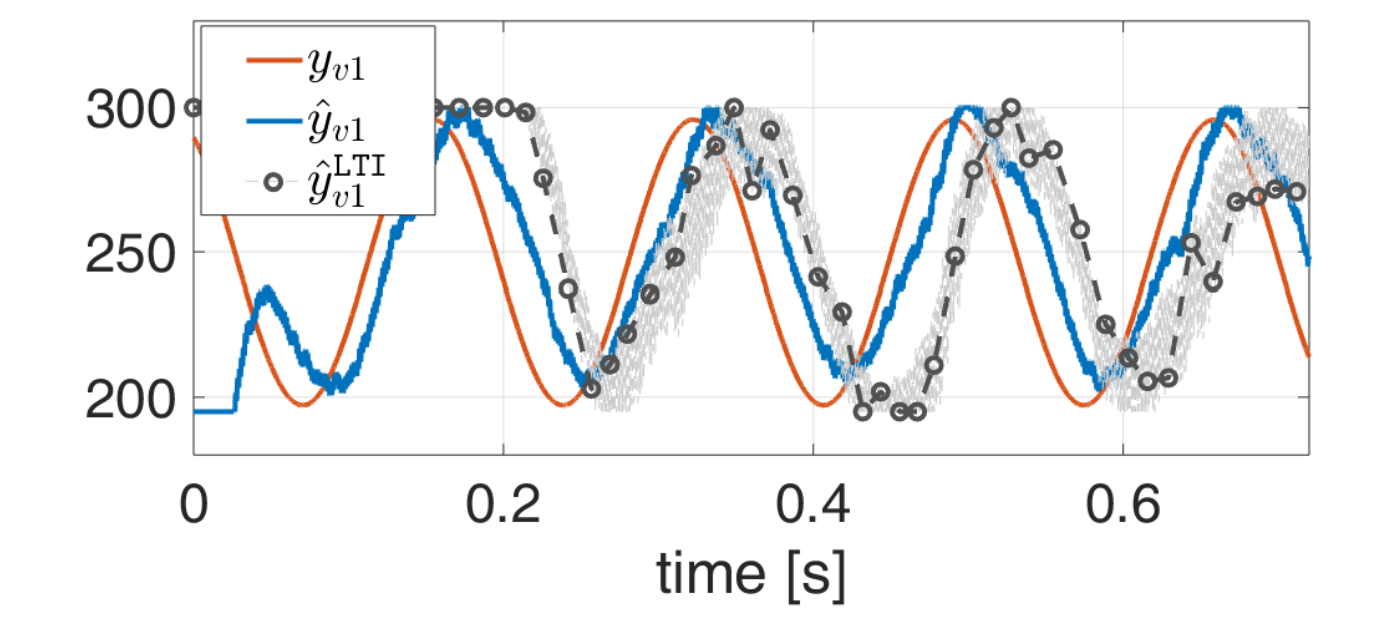}
    \label{fig:comp2}
}
\caption{{Comparison between the proposed estimator and the conventional LTI filter via experiments}}
\end{figure}

%

\begin{figure}[h]
\centering
\subfigure[Test frequency 200 Hz]{
    \includegraphics[width=0.22\textwidth,height=0.15\textwidth]{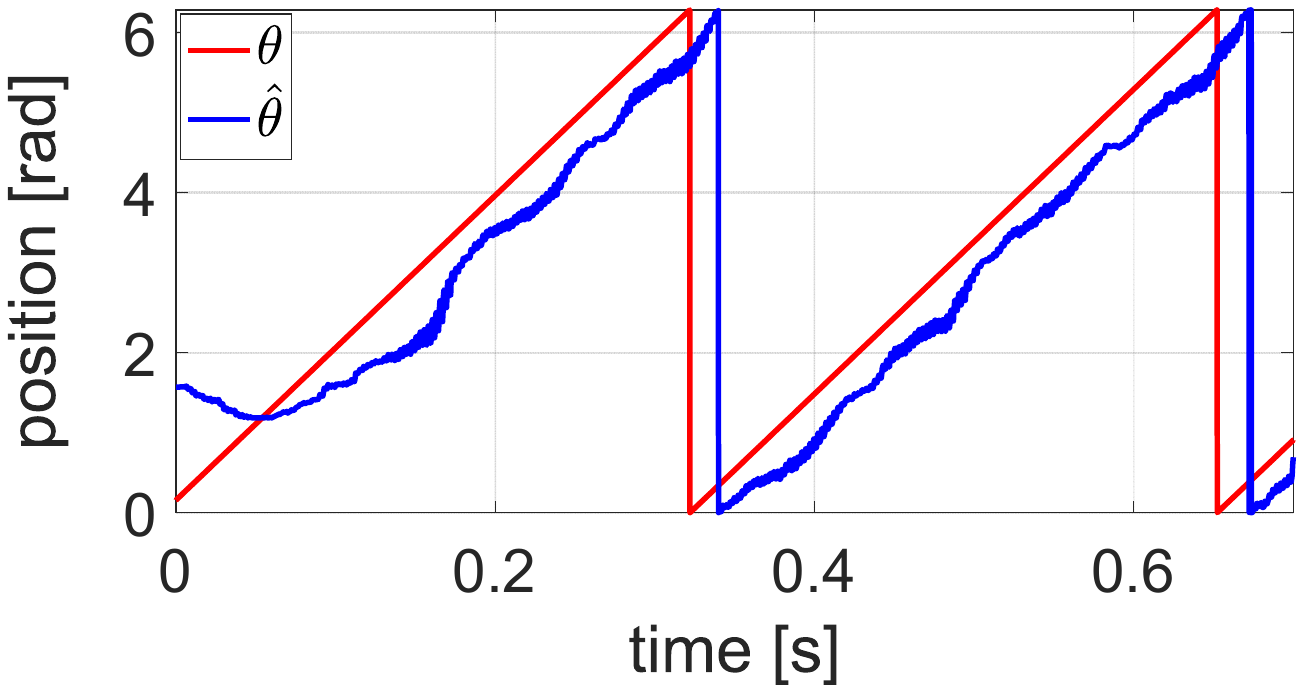}
    \label{fig:exp200Hz}
}
\subfigure[Test frequency 100 Hz]{
    \includegraphics[width=0.22\textwidth,height=0.15\textwidth]{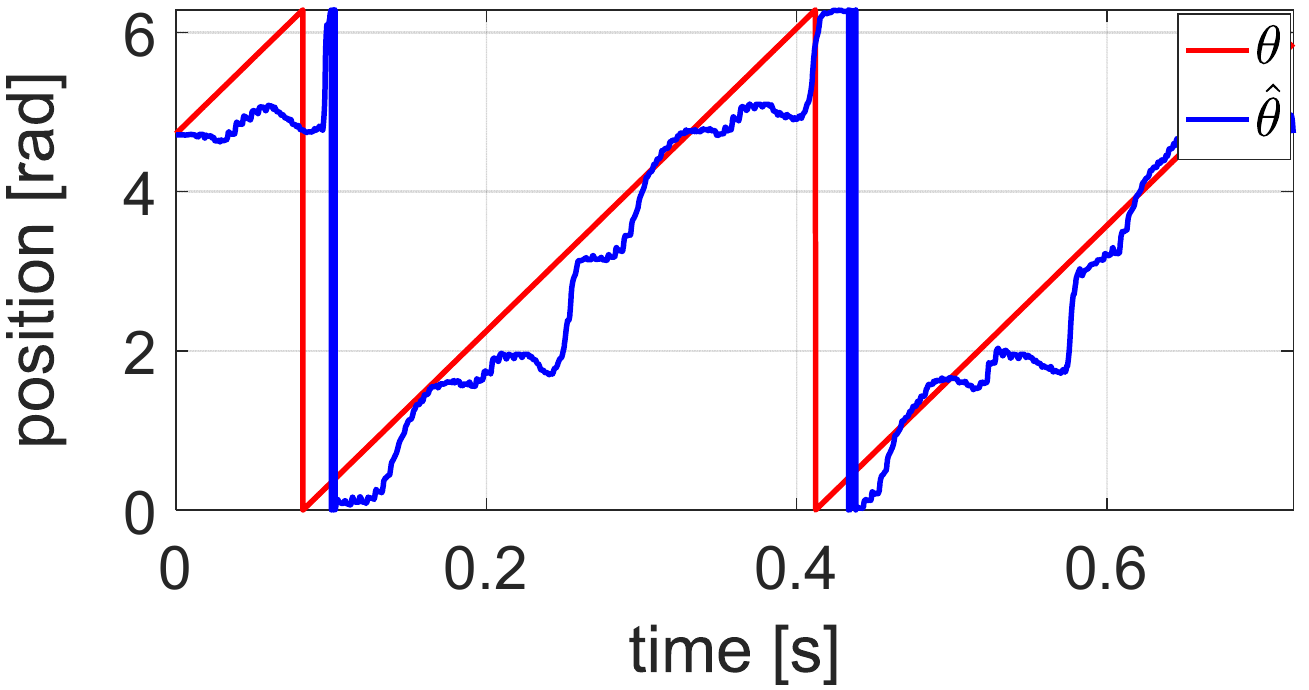}
    \label{fig:exp100Hz}
}
\caption{Effects of the injection frequency}
\label{fig:frequency}
\end{figure}


{In order to evaluate the proposed HPF, we consider the cases with non-zero reference current $i_q^d$, which is proportional to the load $T_L$, with angular velocities at 60 RPM and 40 RPM, respectively. The corresponding results with constant loads are illustrated in Figs. \ref{fig:exp_60} and \ref{fig:exp_40}, in which we get relatively satisfactory performances. We then test the performance with time-varying loads at 40 RPM in Fig. \ref{fig:varying_load}, where a slight distortion may be observed probably due to a relatively heavy load.\footnote{\rm In order to further improve estimation performance, magnetic saturations should be taken into account, but it is out of scope of this paper.} Fig. \ref{fig:reversal} shows the results when the motor under a constant load is {reversing from +20 RPM to -20 RPM in around one second}. It is sufficient to verify the effectiveness of the proposed method in several demanding conditions. }

\begin{figure}[h]
\centering
\subfigure[$\omega = 60$ RPM with constant load ($i_q^d \approx 1$ A)]{
    \includegraphics[width=0.22\textwidth,height=0.15\textwidth]{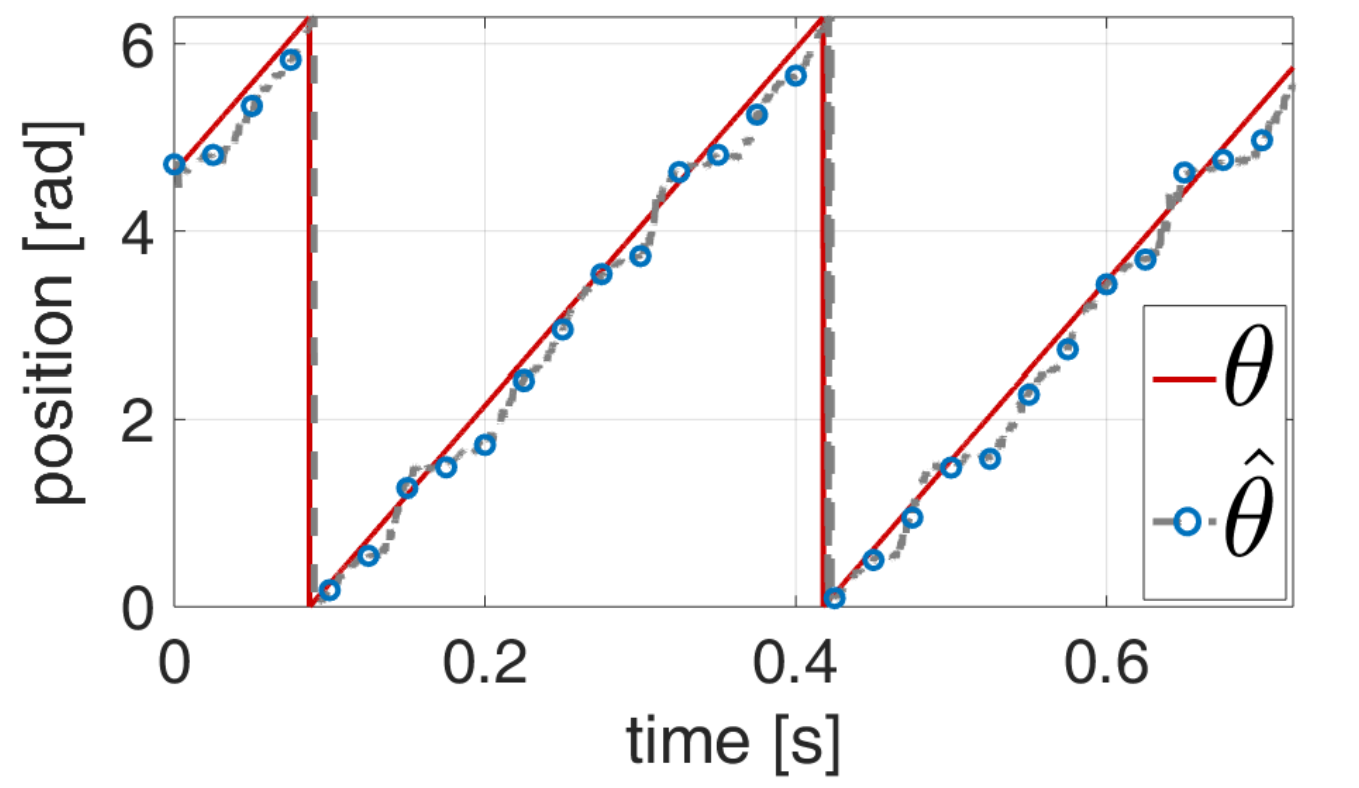}
    \label{fig:exp_60}
}
\subfigure[$\omega = 40$ RPM with constant load ($i_q^d \approx 1$ A)]{
    \includegraphics[width=0.22\textwidth,height=0.15\textwidth]{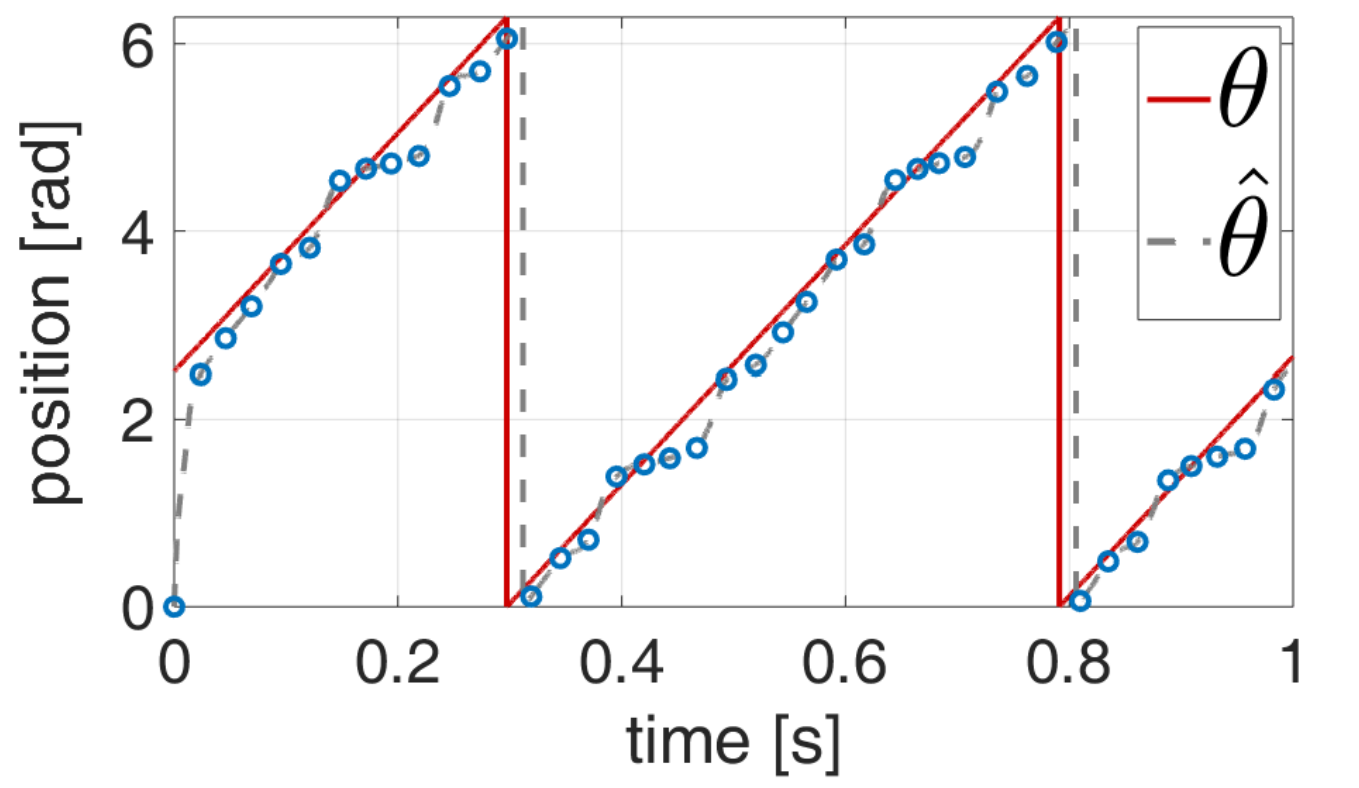}
    \label{fig:exp_40}
}
\subfigure[Time-varying load at $\omega = 40$ RPM]{
    \includegraphics[width=0.22\textwidth,height=0.15\textwidth]{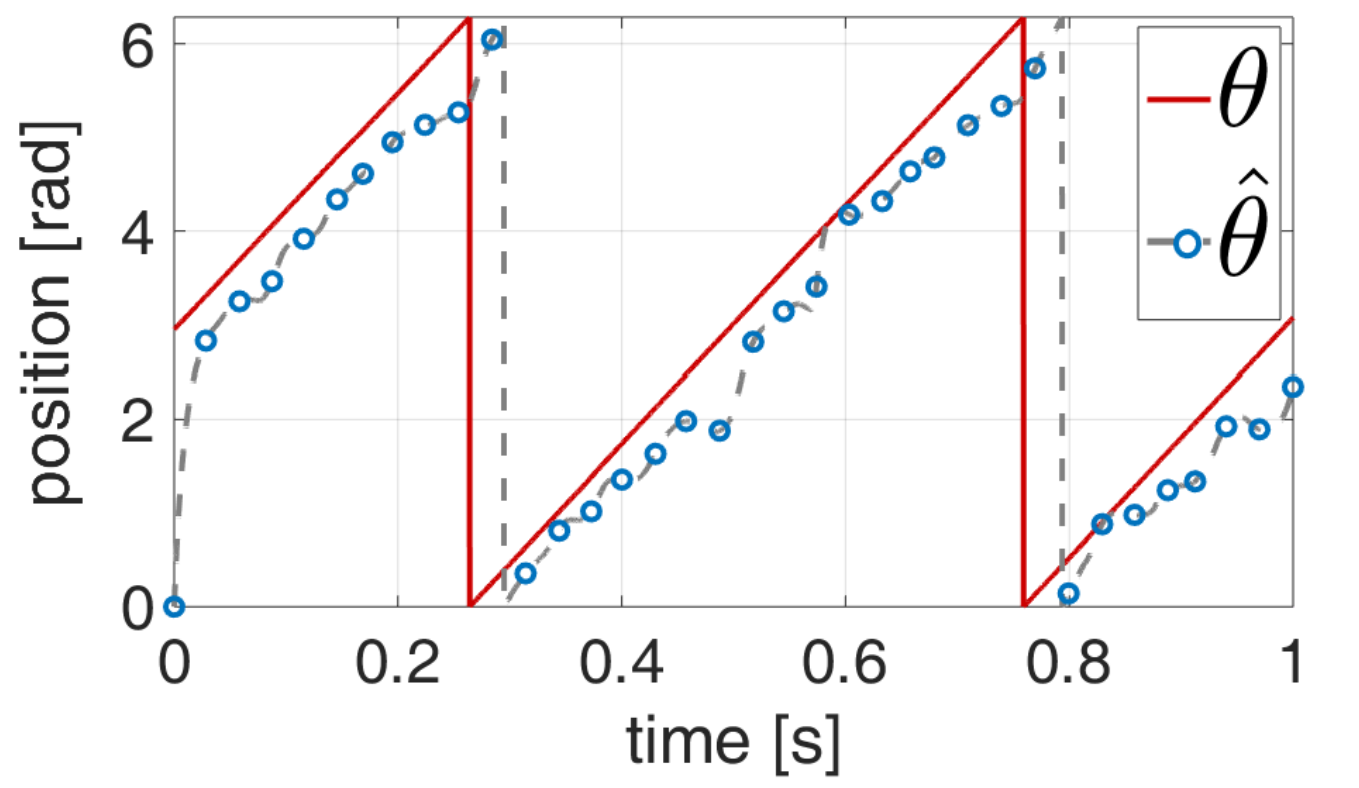}
    \includegraphics[width=0.23\textwidth,height=0.15\textwidth]{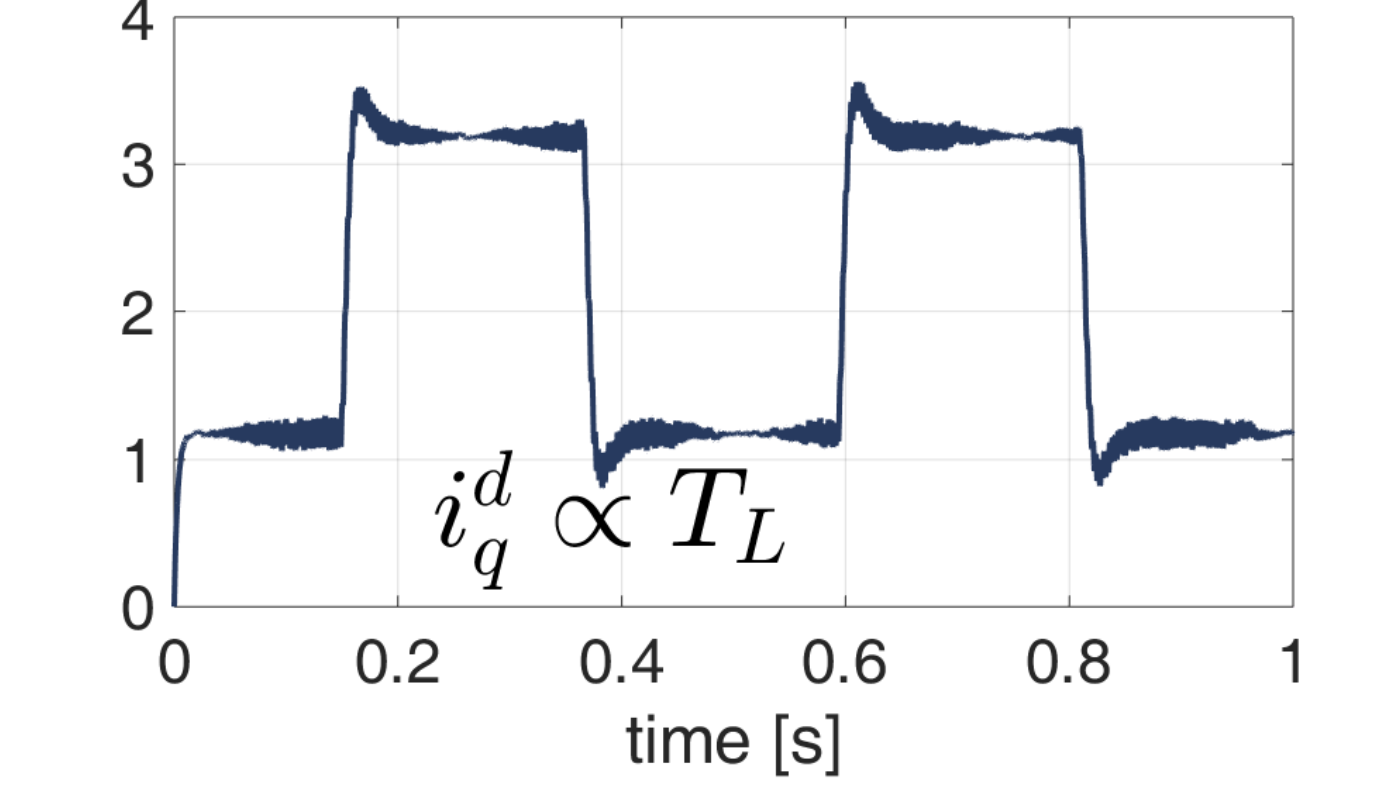}
    \label{fig:varying_load}
}
\caption{{Estimation performance with loads}}
\label{fig:ctrl_exp2}
\end{figure}

\begin{figure}[h]
\centering
    \includegraphics[width=0.45\textwidth,height=0.18\textwidth]{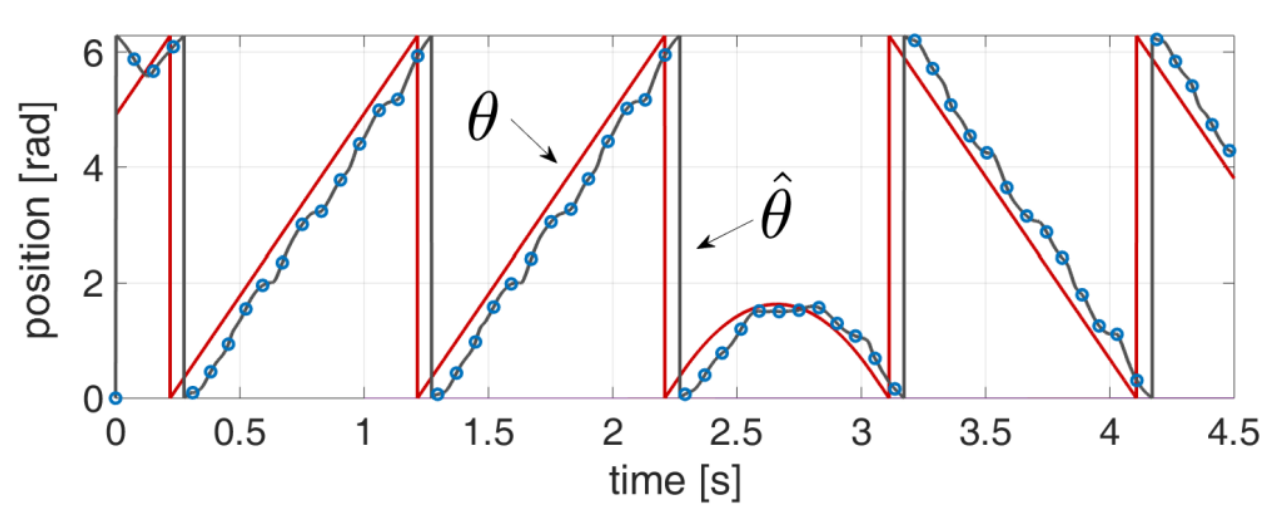}
    \includegraphics[width=0.45\textwidth,height=0.18\textwidth]{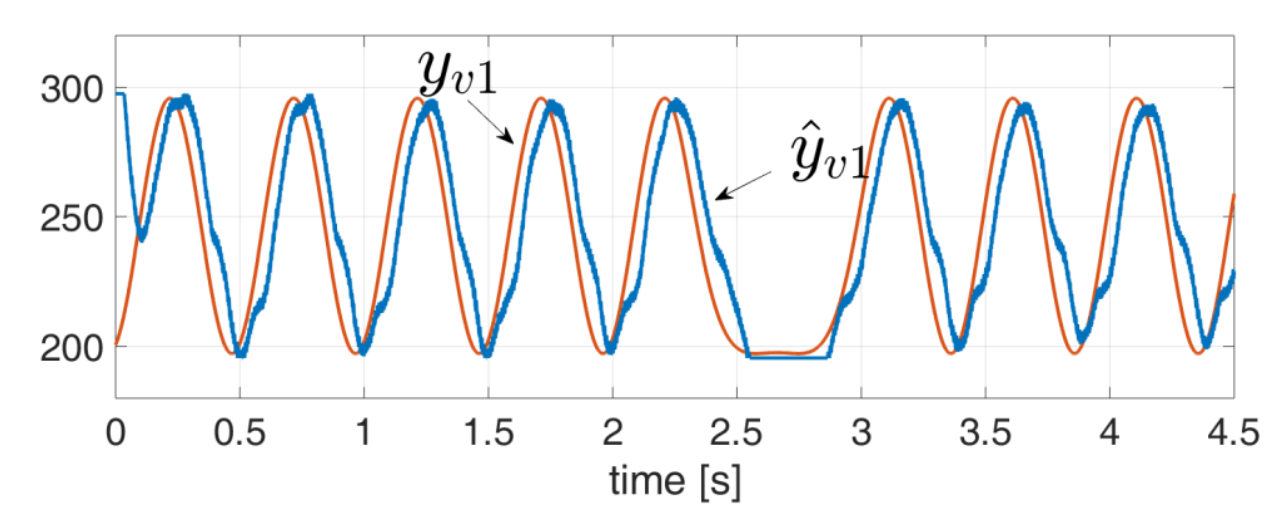}
\caption{{Speed reversal from +20 RPM to -20 RPM ($i_q^d \approx 1.8$ A)}}
\label{fig:reversal}
\end{figure}

\section{\cfm{\large Conclusion}}
\label{sec7}
%
This paper addresses the problem of position estimation of IPMSMs at low speeds and standstill. Although the saliency-tracking-based methods are effective and widely-studied, the theoretical analysis of the conventional methods, taking into account the nonlinear dynamics of IPMSMs, was conspicuous by its absence. This paper attempts to fill in this gap analysing the stator current $\iab$ via the averaging method, with guaranteed error with respect to the injection frequency $\omega_h$. Also, with the key identity \eqref{ident1}, we develop a new position estimator, which ensures an improved accuracy. Moreover, we establish the connection between the new method and the conventional one, showing that they can be unified in the HPF/LPF framework from the perspective of signal processing.

The following extensions and issues are of interest to be further explored.
\begin{itemize}
  \item For the sake of clarity, we only study the basic case of signal-injection methods for the IPMSM model \eqref{volt_model}. The proposed method can also be extended to other motor models, for instance, saturated interior (or surface mounted) PMSMs.

      \vspace{0.1cm}

  \item It is of interest to couple the proposed method with some model-based (non-invasive) techniques, for instance the gradient descent observer in \cite{ORTetalcst,MALetal}, in order to be able to operate the sensorless controller over a wide speed range. Such an approach has been pursued in \cite{CHOetal,ORTetalauto}.

   \vspace{0.1cm}

  \item {The proposed method is relatively sensitive to power converters dead times. It is of practical interests to develop a self-commissioning methods by means of adaptive observers, which is promising to obtain a more practically useful scheme.}
\end{itemize}

\section*{\cfm{\large Acknowledgment}}

The authors would like to thank Fran\c{c}ois Malrait at Schneider Electric for some technical clarifications. This paper is supported by the National Natural Science Foundation of China (61473183, U1509211, 61627810), National Key R\&D Program of China (2017YFE0128500), China Scholarship Council, and by the Government of the Russian Federation (074U01), the Ministry of Education and Science of Russian Federation (14.Z50.31.0031, goszadanie no. 8.8885.2017/8.9).


\end{document}